\documentclass[11pt, a4paper]{article}
\usepackage{titlesec,titletoc}

\usepackage[margin=1.2in]{geometry}

\usepackage[
			anchorcolor=black,
			citecolor=black,
			colorlinks=true,
			filecolor=black,
			linkcolor=black,
			menucolor=black,
			runcolor=black,
			urlcolor=black,
			hyperfootnotes=false
]{hyperref}
\usepackage{amsthm,amsmath,microtype}
\usepackage[noabbrev]{cleveref}

\theoremstyle{plain}
\newtheorem{lemma}{Lemma}
\newtheorem{proposition}{Proposition}
\newtheorem{corollary}{Corollary}
\theoremstyle{definition}

\newtheorem{assumption}{Assumption}

\newtheorem{example}{Example}
\theoremstyle{remark}

\crefname{observation}{Observation}{Observations}
\crefname{theorem}{Theorem}{Theorems}
\crefname{lemma}{Lemma}{Lemmata}
\crefname{proposition}{Proposition}{Propositions}
\crefname{corollary}{Corollary}{Corollaries}
\crefname{definition}{Definition}{Definitions}
\crefname{assumption}{Assumption}{Assumptions}
\crefname{section}{Section}{Sections}
\crefname{figure}{Figure}{Figures}
\crefname{axiom}{Axiom}{Axioms}
\crefname{modification}{Modification}{Modifications}
\crefname{example}{Example}{Examples}
\crefname{appendix}{Appendix}{Appendices}

\usepackage[font={small}]{caption}[2008/08/24] 

\usepackage{enumitem} 

\usepackage{setspace}
 \usepackage{amsmath,amsfonts,amssymb, mathrsfs}
\usepackage[T1]{fontenc}
\usepackage[]{lmodern}
\usepackage[utf8]{inputenc}
\usepackage[ngerman, english]{babel}
\usepackage{csquotes} \usepackage{etex} \usepackage[subrefformat=parens,labelformat=parens]{subfig}
\usepackage{multirow,array}

\usepackage[noabbrev]{cleveref}

\usepackage{pgfplots}
\pgfplotsset{compat=newest}
\usepackage{sidecap} \usetikzlibrary{patterns,decorations.pathreplacing,positioning,external,shapes.geometric,arrows,calc}

\usepackage{pstricks} 
\usepackage{pst-3d}
\usepackage{sgamevar, egameps} \usepackage{graphicx}
\usepackage{epstopdf}

\usepackage[backend=biber,
			sorting=ynt,
			sortcites=true,
			bibencoding=inputenc,
			bibstyle=authoryear -comp,
			citestyle=authoryear -comp,
			doi=false,
			giveninits=true,
			isbn=false,
			uniquename=false,
			uniquelist=minyear,
			maxcitenames=3,
			maxnames=3,
			maxbibnames=99,
			natbib=true,
			uniquename=false,
			url=false,
			date=year
]{biblatex}
\bibliography{aip_flat_biber}

\DeclareSourcemap{
  \maps[datatype=bibtex]{
    \map[overwrite=true]{
      \step[fieldsource=number]
      \step[fieldset=issue, origfieldval]
    }
  }
}

	\AtEveryBibitem{\clearlist{language}}
	\AtEveryBibitem{\clearfield{number}}
	\AtEveryBibitem{\clearlist{month}}
	\AtEveryBibitem{\clearfield{Month}}
	\AtEveryBibitem{\clearfield{day}}
	\AtEveryBibitem{\clearfield{eprint}}

\renewbibmacro{in:}{}
\usepackage{verbatim}
\usepackage{booktabs} \usepackage[multiple]{footmisc}

\usepackage{placeins} 

\definecolor{Vermilion}{RGB}{213,94,0}
\definecolor{BluishGreen}{RGB}{0,158,115}
\definecolor{Cadmium}{RGB}{227, 0, 34}
\definecolor{ReddishPurple}{RGB}{204,121,167}
\definecolor{Orange}{RGB}{230,159,0}
\definecolor{Yellow}{RGB}{240,228,66}
\definecolor{SkyBlue}{RGB}{86,180,233}
\definecolor{Blue}{RGB}{0,114,178}

\DeclareFontFamily{U}{mathx}{\hyphenchar\font45}
\DeclareFontShape{U}{mathx}{m}{n}{
      <5> <6> <7> <8> <9> <10>
      <10.95> <12> <14.4> <17.28> <20.74> <24.88>
      mathx10
      }{}
\DeclareSymbolFont{mathx}{U}{mathx}{m}{n}
\DeclareFontSubstitution{U}{mathx}{m}{n}
\DeclareMathAccent{\widecheck}{0}{mathx}{"71}
\DeclareMathAccent{\wideparen}{0}{mathx}{"75}

 \usepackage{soul}
\setlist{noitemsep,topsep=1pt}
\assignrefcontextentries[]{*}
\usepackage{tabularx}
\title{On Risk and Time Pressure: \\When to Think and When to Do\thanks{We thank B\aa rd H\aa rstad and three anonymous referees for their valuable input that discretely improved the paper. Sarah Auster, Laura Doval, William Fuchs, Heiko Karle, Nenad Kos, Aditya Kuvalekar, Antoine Loerper, Matt Mitchell, Marco Ottaviani, Ronny Razin, Emanuele Tarantino, and seminar audiences at Bocconi, Bergamo IO Theory Workshop, Collegio Carlo Alberto, Duke, SAEe Madrid, SAET Seoul, MaCCI Annual Conference Mannheim, Oslo, SIOE Stockholm, and uc3m provided helpful comments. Johannes Schneider gratefully acknowledges financial support from the German Research Foundation (DFG) through CRC TR 224 (Project B03), Agencia Estatal de Investigación (grant PID2019-111095RB-I00 and grant PID2020-118022GB-I00), Ministerio Economía y Competitividad (grant ECO2017-87769-P), and Comunidad de Madrid (grant MAD-ECON-POL-CM H2019/HUM-5891).}}

\makeatletter
\def\pgfplots@stacked@diff{}
\makeatother

\author{Christoph Carnehl\thanks{Bocconi University, Department of Economics and IGIER; E-mail: \nolinkurl{christoph.wolf@unibocconi.it}} \and 
Johannes Schneider\thanks{University of Mannheim and Carlos III de Madrid; E-mail: \nolinkurl{jschneid@econ.uc3m.es}}
}

\date{\today}
\begin{document}
\maketitle 

\begin{abstract}
We study the tradeoff between fundamental risk and time. A time-constrained agent has to solve a problem. She dynamically allocates effort between implementing a risky initial idea and exploring alternatives. Discovering an alternative implies progress that has to be converted to a solution. As time runs out, the chances of converting it in time shrink. We show that the agent may return to the initial idea after having left it in the past to explore alternatives. Our model helps explain so-called false starts. To finish fast, the agent delays exploring alternatives reducing the overall success probability. \end{abstract}
\newpage
\setstretch{1.25}

\begin{quote}
The people that really create the things that change this industry are both the 'thinker-doer' in one person.\\
\textemdash \emph{Steve Jobs. \citet{Jobs}.}
\end{quote}

\section{Introduction} \label{sec:introduction}
Problem solving under time pressure is central to many economic problems. For example, consider an entrepreneur who has received funding for her venture. While the initial funds cover her expenses to advance the project for some time, she requires additional financing rounds before the venture becomes profitable.\footnote{\citet{gompers1995} discusses the importance of staged financing for venture capital-backed startups. The empirical analysis highlights that financing rounds are short. On average, they last just above one year.} To raise new funds, the entrepreneur has to achieve a milestone: build a prototype, solve a technological problem, or prove that a promising market for the product exists. Often, the entrepreneur faces an explicit or implicit deadline to achieve the milestone.\footnote{Running out of cash is the second most frequent reason for startup failure according to \citet{whystartupsfail}. The entrepreneur has to achieve the milestone before her funds run out: an implicit deadline. Moreover, \citet{kaplan2003} documents that venture capitalists use both ex ante staging\textemdash committing to milestones that have to be achieved by a deadline\textemdash and ex post staging\textemdash liquidating the venture if the entrepreneur's  performance is not satisfactory when the new funding round is due.}

Suppose that the entrepreneur proposed a particular strategy to reach a milestone\textemdash for example, a successful launch in a particular market. While the strategy appears promising, there is a risk that it is fundamentally flawed. However, there may be alternative ways to prove the venture's worth to investors. Invoking such an alternative would imply that the entrepreneur \emph{pivots} away from her initial idea to a new strategy.

A successful pivot requires preparation. First, stakeholders need to be convinced that the change in strategy is promising \citep{PivotStory}. To this end, the entrepreneur could, for example, conduct customer research to discover market needs, leading to a new business strategy. When contemplating the option to pivot, our entrepreneur faces the following tradeoff \citep{NotPivot}: While preparing a pivot requires scarce time and resources, continuing with the current strategy entails the risk of failing.

When should the entrepreneur prepare a pivot? Early on, when time pressure is still low? Or later, after she has experimented unsuccessfully with the initial idea? When---if ever---should she abort the preparations for the pivot and set focus on the initial idea again?

In this paper, we study the novel tradeoff between fundamental risk and time in a multiarmed, continuous-time bandit model. We consider an agent (e.g., our entrepreneur) that has to solve a problem (e.g., achieve a milestone) by a deadline. Following the Steve Jobs quote at the beginning, the agent can exert effort in two different ways: She can \emph{do}, or she can \emph{think}.\footnote{\citet{Bolton2009Thinking} consider a related but different tradeoff between thinking and doing. Importantly, and different from us, their tradeoff does not vary over time.} Doing corresponds to working on implementing her initial idea. Thinking corresponds to preparing a pivot. While doing, a \emph{solution} arrives stochastically at an ex ante unknown rate. Instead, while thinking, \emph{progress} arrives stochastically at a known rate. While a solution results in a fixed payment (e.g., the next round of funding), progress needs to be converted first. Because converting progress is harder with little time remaining, the value of progress is lower the closer the deadline is.\footnote{There are several ways to model this decreasing value. Our baseline model depicts the value of progress in an abstract reduced form. In \Cref{ssub:examples}, we provide several examples to microfound this reduced form. These include, for example, safe and risky exponential bandit arms or an Ornstein-Uhlenbeck payoff stream.}

A large body of literature in economics studies problem solving as the choice between solution methods. Following \citet{weitzman1979optimal,Rothschild:74}, this literature focuses on an exogenous set of known methods chosen over an infinite horizon.\footnote{In \citet{Rothschild:74}'s language, a method is an arm. In \citet{weitzman1979optimal}'s language, a method is a box.} 

In reality, however, deadlines matter. \citet{Kirtley2021} document, for example, the factors driving entrepreneurs to prepare a pivot. They show that both time pressure and beliefs about the feasibility of the initial idea play a crucial role in an entrepreneur's decisions. Therefore, understanding entrepreneurial choices on pivots in light of time pressure and risk requires explicit modeling of both a deadline and the generation of new strategies. In particular, different from the canonical experimentation approach, our model captures the notion that the value of progress depends on the time left to make use of it.

Our first contribution is that we characterize the agent's optimal policy in such a model. The characterization is nontrivial. Due to the shortening time window, the value of progress changes over time independent of the agent's actions. In such a model, the optimal myopic policy need not be dynamically optimal, and the index theorem of \citet{gittens1974dynamic} does not apply.

If the agent is optimistic about her initial idea and the time window is large, then the optimal policy is as follows. The agent starts by doing. If doing remains unsuccessful, then the agent switches to thinking. However, if progress does not arrive in due time, then the agent switches back to doing and, in the remainder of the time, aims for a solution via the initial idea.

The property that the agent returns to a previously discarded arm results from three model ingredients: the finite horizon, the fact that progress needs to be converted in a second step, and the positive cost of effort. When any of these ingredients are dropped, the optimal strategy becomes a classical one-time switching rule. If the arm that the agent started with is not successful in due time, then she switches to the other but never switches back. To develop an intuition for the incentives at play in the optimal policy, we separately study two benchmarks\textemdash no time pressure and zero cost of effort.

If we drop time pressure, the problem becomes a standard, recursive infinite-horizon problem.\footnote{
The following benchmark is analogous: Consider a classical experimentation problem with only one arrival needed on both arms but with the thinking arm having a lower intensity rate. Both benchmarks remove the time-varying value of progress that derives from its multistage nature and the finite horizon.} Thinking becomes equivalent to the safe option. Doing, the risky option, offers a cost advantage. Therefore, the agent first approaches the problem through doing. If no solution arrives, she becomes pessimistic about the quality of her initial idea. She switches to thinking. Because time pressure is absent and the time to convert progress never runs out, the value of thinking is constant. When the agent finds it optimal to switch to thinking at one point, the thinking arm dominates the doing arm for the remainder of the time.

If we drop the cost of effort but keep the time window finite, then the result reverses. Thinking early on has a higher value than thinking later. Early progress leaves ample time for conversion. The agent starts by thinking. If progress remains absent, then she becomes pessimistic about having sufficient time left to convert progress should it arrive. She switches to doing in the hope of an immediate solution.

Our second contribution is to use our model to explain some entrepreneurial decision-making peculiarities that traditional bandit models cannot explain. First, we provide a theoretical rationale for \emph{false starts}. A false start describes the entrepreneur's tendency to act on initial ideas (doing) rather than to invest in customer research to explore alternatives (thinking). False starts are costly, because the earlier customer research is done, the more time and resources remain to improve the initial idea.  \citet{whystartupsfail2} argues that false starts are one of the main reasons startups fail. Indeed, following our model, entrepreneurs explore alternatives too late, leaving them relatively little time to convert a promising pivot into a successful venture. In return, if they succeed with their initial idea, they succeed early. In our model, by doing early, entrepreneurs trade off overall success probability against saving time and the resources needed for early thinking.\footnote{One case in point is the Triangulate venture, as discussed in \citet{whystartupsfail2}. In a post-mortem of his failure, the founder Sunil Nagaraj admitted that he rushed to launch the venture’s platform Wings rather than spending time on customer research to verify the market need for an improved matching engine. He considers that behavior as one of the main mistakes leading to the eventual failure of Triangulate.}

Second, an increase in the ex ante belief about the initial idea can lower the overall success probability of the entrepreneur. A higher initial belief discourages early customer research, thereby amplifying the false-starts problem. If the deadline is not too short, this effect may dominate the positive effect of a higher likelihood that the initial idea can deliver a solution.

Third, \emph{increasing} time pressure can incentivize the agent to think early on. A venture capitalist can mitigate the false start problem through tighter deadlines that crowd out the entrepreneur's incentive to do early. In general, early doing is attractive because it can provide a quick solution. The entrepreneur may opt for this route if, initially, the time pressure is not too high. By increasing the initial time pressure, the venture capitalist discourages early acting on the initial idea. Instead, it encourages customer research when it is most valuable.

We believe that our modeling framework also applies in other contexts. While entrepreneurial problem solving serves as our main application, we discuss other applications in \cref{sec:conclusion}, our \nameref{sec:conclusion}.

\paragraph{Related Literature.} We contribute to a large body of literature that deals with the choice between approaches to innovation. One strand of the literature dating back to \citet{weitzman1979optimal} has considered several variants of Pandora's box problem as a proxy for finding the right innovation strategy \citep[e.g.,][]{fershtman1997simple,doval2018whether,olszewski2015more}. Other works have been concerned with how competition affects the search for the right approach \citep[e.g.,][]{lemus2019diversification,letina2016road,akcigit2015role,10.2307/2695893}. Our contribution to this literature is that we endogenize the available approaches by allowing the agent to explore an alternative route. Therefore, we also endogenize the cost of finding an alternative. We capture two aspects absent in the aforementioned literature. First, new ideas arrive stochastically, and the cost the agent incurs to make progress on the alternative varies with the time it takes until progress arrives. Second, and more importantly, the value of progress varies with the time window the agent has to convert progress into a solution. Thus, the \emph{availability} of an alternative route depends on both luck and choices in the past, and the \emph{value of discovering} an alternative depends on the time window left.\footnote{There is a literature that studies the choice between risky, innovative approaches \citep[see, for example,][]{chen2018patentability,das2019inefficient} compared to a safe and established alternative. While the choice set is given in these models, our focus is on the \emph{search} for better alternatives because the existing approach carries some risk. Therefore, in our model successful search mitigates the agent's risk.}

A strand of the management literature addresses issues similar to ours. An example is \citet{doi:10.1002/smj.3010}. However, their \emph{Test Two, Choose One} result ignores the time dimension, which is the focus of our paper. The process of how to think about alternatives and the particularities of lean techniques are discussed in \citet{FELIN2019101889}. Our model provides a formal, economic method for these ideas.

Technically, our model falls into the class of multiarmed bandit problems \citep{Rothschild:74}. \citet{bergemann2008bandit} provides an overview of the literature. The doing arm is a classical continuous-time exponential bandit, as used in most of the strategic experimentation literature \citep{KRC05}. While most models feature infinite-horizon settings, we are interested in a time-constrained agent. \citet{klein2016importance} also considers a time-constrained agent. The crucial difference from our model is that, in his case, both arms are exponential bandits that differ in their intensities, not in the number of arrivals needed.

In our case, the thinking arm is \emph{restless}\textemdash the state of the arm evolves even when not pulled. The restless feature is essential for capturing the risk-time tradeoff but makes the model complicated \citep[see, e.g.,][]{68001}. The thinking arm concept is related to the few papers in the literature that study multistage bandits \citep{wolf2017informative,green2017,hu2014financing,moroni2015experimentation,KELLER2003branching,StudyGoForIt}.

The results in \citet{StudyGoForIt} appear reminiscent of ours, yet the environment and thus the mechanism differ substantially. We study an agent’s decision whether to experiment with a risky project or take a step back and look for a safer alternative absent agency concerns. \Citet{StudyGoForIt} instead abstracts from risk and focuses on the agency problem alone. In his model, multistage projects benefit an infinitely lived principal, as intermediate reports provide a monitoring tool. Absent agency concerns, multistage projects are inefficient and are never chosen. In our setting, it is a priori unknown which approach is efficient, and the agent balances learning and managing time pressure.\footnote{The role for deadlines in \citet{StudyGoForIt} is therefore closer to that established in \citet{bonatti2011collaborating}. A deadline provides an instrument to incentivize the agent to work. In our model, instead, deadlines affect the agent's tradeoff between improving the likelihood of finding a solution at all and---conditional on making progress---having enough time to benefit from it.} 

\citet{10.1093/restud/rdw008,10.2307/23045643,nikandrova2018dynamic,10.1257/aer.20171000,Francetich2019} study the problem of dynamically distributing effort across several projects. However, none of them addresses the risk-time tradeoff. \citet{Francetich2019} studies the choice of allocating effort to two correlated bandits, \citet{10.2307/23045643} features myopic agents, and \citet{10.1093/restud/rdw008} considers a two-period overlapping-generations model. By construction, neither reproduces the switching dynamics we obtain. \citet{nikandrova2018dynamic} model an agent who irreversibly selects between two alternatives. The agent uses experimentation to learn about her options beforehand. In \Citet{10.1257/aer.20171000}, an agent can try to find a solution or to show that no solution exists. At the optimum, the agent uses only one of the available routes until she reaches an absorbing state. All five models are related in spirit, but the research question, the modeling choices, and the results are different.

The two most closely related papers to ours are \citet{Bolton2009Thinking,FershtmanPavan}. Most notably, however, both of their models operate with an infinite time horizon and thus cannot produce the time-pressure dynamics of our finite-horizon problem. In their models, the tradeoff between arms is independent of calendar time. The available time window and associated time pressure constantly change in our model, affecting the main tradeoff.

\citet{Bolton2009Thinking} are also interested in the choice between thinking and doing. However, thinking plays a different role in their model. It is a tool to resolve uncertainties regarding future choices. In contrast, thinking in our model corresponds to the development of a new, previously unavailable route.

Similar to our work, \citet{FershtmanPavan} study a model with endogenous arms. They consider an agent with an infinite horizon who decides whether to apply her initial idea or to search for alternatives. Their focus is on the search process itself. The agent can, at a cost, investigate several routes and learn about their quality. Instead, we focus on the risk-time tradeoff that the agent faces as the available time window closes. To gain tractability in a finite-horizon world, we abstract from some details of the search process. We collapse it into a unidimensional object. Because different approaches to the problem are taken, the results also differ. While \citet{FershtmanPavan} show that in their setting, an index policy remains optimal, we show that the same does not hold when the risk-time tradeoff plays a role.

\paragraph{Roadmap.} We set up our model in \cref{sec:model}. We derive our main results in \cref{sec:analysis}. \cref{sec:finalremarks} describes a set of economic implications derived from our findings. In \cref{sec:discussion}, we discuss our modeling choices. Finally, \cref{sec:conclusion} concludes and provides an outlook on other applications of our framework.

\section{Model}
\label{sec:model}

We introduce a model capturing the risk-time tradeoff between solving a problem through \emph{doing}\textemdash a fast but fundamentally risky approach\textemdash or through \emph{thinking}\textemdash a slower but less risky approach. We provide a discussion of our assumptions in \cref{sec:discussion}.

\paragraph{Setting.} Time is continuous and starts at $t=0$. An agent has to solve a problem by a finite deadline $T<\infty$ and there is no discounting over time. At each instant of time, the agent can invest one divisible unit of effort into \emph{doing}, $a_t^d$, and \emph{thinking}, $a_t^s$, such that $a_t^d+a_t^s\leq 1$. Investing effort entails a flow cost of $(a_t^d+a_t^s)c$, with $c>0$. 

One arm, the \emph{doing} arm, is risky. The arrival of a solution on that arm depends on the unobserved binary state $\theta \in \{0,1\}$. The instantaneous arrival rate of the arm is $\theta \lambda$, where $\lambda>0$, which implies the following: if $\theta=1$, then the probability that a solution arrives when the agent invests $a_t^d$ over a small time interval $[t,t+dt)$ is $\lambda a_t^d dt$; if $\theta=0$, then a solution never arrives on the doing arm. The agent's belief that $\theta=1$ at $t=0$ is $\bar{p} \in (0,1)$. A solution delivers a payoff of $B>0$ to the agent and ends the game.

The other arm, the \emph{thinking} arm, has a known instantaneous arrival rate $\mu>0$. An arrival on the thinking arm at time $t$ implies \emph{progress}. Progress does not provide a solution directly. Instead, the agent has to convert progress into a solution that requires additional time and effort. We capture this second step in reduced form. The function $V(\tau)$ describes the agent's continuation payoff when progress occurs with time $\tau=T-t$ remaining to the deadline. We assume that $V(\cdot)$ is thrice continuously differentiable, increasing, i.e., that $V'(\cdot)>0$, and sufficiently concave, i.e., that $-V''(\cdot)/V'(\cdot)\geq \bar{p} \lambda$. Moreover, an arrival on the thinking arm with no time remaining is worthless, $V(0)=0$. We assume that all derivatives have a limit as $\tau \rightarrow \infty$. Moreover, we assume (abusing notation) $V(\infty):=\lim_{\tau \rightarrow \infty} V(\tau) > c/\mu$ and $\lim_{\tau \rightarrow \infty}-V''(\tau)/V'(\tau) \geq \bar{p} \lambda$.\footnote{The second-to-last condition states that the expected value of progress without time pressure is larger than the expected effort cost required to obtain progress. The last condition strengthens the concavity assumption, assuming it is true in the limit. We need it for a few of our results.}

To streamline the intuition, we interpret progress in our main analysis as the first step in a multistage problem. We imagine that the agent works on converting progress into a solution after progress occurs. We provide detailed examples that microfound $V(\tau)$ for this interpretation and others in \cref{ssub:examples}. However, formally, as with the doing arm, an arrival on the thinking arm ends the game.

Throughout this paper, we are interested in how the agent allocates effort \emph{between} doing and thinking. We focus on cases in which the agent finds it optimal to exert full effort until the end of the game. We thus consider cases in which the solution yields a high reward; that is, we assume that $B$ is sufficiently large. It is straightforward to show that such a reward $B$ exists; \Cref{sub:no_shirking_conditions} provides the respective argument. A large $B$ allows us to restrict attention to the agent's choice \emph{between arms}, as it implies $a_t^d= 1-a_t^s$. Dropping the superscript and denoting time variables in terms of the time remaining $\tau$, we use $a_{\tau}\equiv a_{T-t}^d$.

\subsection[Examples of Progress]{Examples of $V(\tau)$}
\label{ssub:examples}

The crucial feature of the thinking arm is that the value of progress depends on the time remaining until the deadline. To illustrate the model's flexibility in the context of our leading application---startups---, we pause here and provide a set of examples nested in our model. We provide a formal verification that each example meets our assumptions in \cref{sub:ex_verify}.

\subsubsection*{Implementation Phase or Delay upon Progress}
In our first set of examples, progress triggers a new bandit arm that the agent can pull following progress. One interpretation is that an entrepreneur with a minimum viable product can move forward in two ways: Conduct customer research to determine the optimal market for her project, or proceed immediately to launch the product without customer research. In the context of our model, the first option corresponds to pulling the thinking arm, the second to pulling the doing arm. We get back to this interpretation when we discuss our model implications in \cref{sec:finalremarks}.

The simplest way to model this interpretation is given by our leading example, \cref{example_implications}.

\begin{example}\label{example_implications}
An arrival on the doing arm delivers a new Poisson bandit with a known arrival rate $\nu\geq \lambda \overline{p}$. An arrival on this new arm implies a solution worth $B_\nu$ to the agent, and the cost of pulling the arm is $c_\nu$.
\end{example}

\Cref{example_implications} describes the case when successful customer research, i.e., the arrival of progress on the thinking arm, delivers a new product design, i.e., a new approach corresponding to a Poisson bandit arm, that can be marketed with certainty\textemdash given enough time\textemdash with arrival rate $\nu\geq\lambda \overline{p}$. The new product design can either be an improved version of the original idea or an entirely new product inducing a \emph{pivot} of the entrepreneur. We now present a set of related ideas also subsumed by our model.

\begin{example}
\label{par:ex_risky} This example is in a similar vein. However, successful customer research may not remove all uncertainty about the feasibility of a new product idea. Instead, the new product may deliver better prospects than the original product, but some uncertainty remains. The new arm is then characterized by $\bar{p}^\nu$, the ex ante belief about the new arm's feasibility, and by $\nu$, the new arm's arrival rate. An arrival on this arm provides value $B_\nu$ and pulling it has a flow cost of $c_\nu$. \end{example}

\begin{example}\label{par:ex_timevarying} Our model also nests cases in which successful thinking triggers a new arm that has a time-varying intensity rate $\nu(t)$. Successful customer research suggests a product quite different from the original idea and the entrepreneur's experience. The time-varying intensity can go both ways. On the one hand, there can be an increasing intensity: initially, the agent may not be very knowledgeable about the new approach, but her knowledge improves over time. On the other hand, there can be a decreasing intensity: the customer research may demonstrate some immediate implications expected to work with high probability. However, if these initial attempts fail, then it becomes harder to succeed. For concreteness, suppose the intensity rate follows an exponential function $\nu(t)=\nu e^{\alpha+\beta t}$, with $\alpha\geq 0$, $\beta\geq - \frac{1}{\tau}$, and $\nu B>c$.\footnote{While we do not restrict the sign of $\beta$ to allow for both an increasing ($\beta>0$) and a decreasing ($\beta<0$) intensity rate, we impose a lower bound on $\beta\geq - \frac{1}{\tau}$, which implies $\nu(t)B\geq c$ given $\alpha \geq 0$. In particular, $\alpha \geq \ln\left(\frac{c}{\nu B}\right)-\beta \tau$, where the right-hand side is strictly negative for all $\beta\geq 1/\tau$, as $\nu B>c$ by assumption.} \end{example}

\begin{example}\label{par:ex_delay} Finally, consider a model in which customer research provides enough evidence for the entrepreneur to raise an additional round of financing. However, it takes an unknown amount of time to analyze and polish the data, prepare a convincing pitch, or find the right venture capitalist. In such a model, the problem is solved by progress on the thinking arm, but the payoff is realized only after some random delay, which follows a Poisson process with intensity rate $\nu\geq \bar{p} \lambda$. Suppose that the entrepreneur has to incur a flow cost of $c_\nu$ until its arrival (for example, to pay her employees). Moreover, the project's payoff is $B_\nu$ if the success arrives before the deadline and $0$ if not. This example is formally equivalent to \cref{example_implications}.\end{example}

\paragraph{Why Abandon the Doing Arm upon Progress?} In the above examples, the entrepreneur abandons the doing arm upon progress on the thinking arm. While this is an endogenous outcome in \cref{example_implications} and \cref{par:ex_delay}, we assume it for \cref{par:ex_risky} and \cref{par:ex_timevarying}. There are several ways to microfound this assumption:
(i) The new arm triggered by progress \emph{replaces} the old arm. (ii) The belief about the new arm is sufficiently high that the agent will never return to the old arm. (iii) The agent would have to pay a maintenance cost to hold the old arm idle \citep[as modeled in][]{forand2015keeping}. (iv) The agent incurs a sufficiently high switching cost when returning to the old arm.

\subsubsection*{Payoff Stream upon Progress}

An alternative model is one in which successful thinking triggers a payoff stream. For example, instead of attempting to build a new prototype with advanced technology (i.e., continuing to pull the doing arm), the entrepreneur finds a way to market the product as is without further improvements. This market opportunity generates a revenue stream until the deadline. A possible interpretation is an entrepreneur evaluated under the scorecard method.\footnote{See, for example, https://www.forbes.com/sites/mariannehudson/2016/01/27/scorecard-helps-angels-value-early-stage-companies/?sh=4e8eb9c96874.} The entrepreneur can either try to improve the product (pulling the doing arm) or bring a new product to the market to increase sales (pulling the thinking arm to find the right way to launch the product and collecting the revenue stream upon progress).

\begin{example}\label{par:ex_stream} We model the launch of a product as paying out $db(t)$ in the interval $[t, t+dt)$ and assume that the payments follow an Ornstein-Uhlenbeck process:\footnote{To keep the notation simple, here, we reset time to 0 once the market opportunity is used.} $db(t)=\nu(B_\nu-b(t)) dt+\sigma dW_t$, where $W_t$ is a standard Brownian motion and the initial value is $b(0)=0$. The payoff $db(t)$ can be thought of as a flow profit that has $B_\nu$ as its long-run expectation. $\nu\geq \bar{p}\lambda$ is the rate of mean reversion of the profit process.\end{example}

\section{Analysis} \label{sec:analysis}
In this part, we characterize the agent's optimal policy. We begin by describing two benchmark results. After that, we provide an interpretable necessary condition for the agent's optimal policy. Finally, we derive an algorithm that characterizes the unique optimal solution under a mild technical assumption.

\subsection{Benchmarks}
\label{sub:benchmarks}

The thinking arm has two essential features that distinguish it from the doing arm: (i) it becomes increasingly unattractive as the deadline approaches, and (ii) its expected payoff is different from that of the doing arm. The source of the first feature lies in the conversion of progress into a solution. Completing the additional steps in time becomes increasingly unlikely when less time remains. The source of the second feature is twofold. First, the expected effort cost until progress arrives, $c/\mu$, may be different from the expected effort cost until a solution of a good ($\theta=1$) doing arm arrives, $c/\lambda$. Second, the value of progress\textemdash even without time pressure\textemdash may be different from the value of a solution on the doing arm, $V(\infty)\neq B$. A reason for the latter is that the agent has to exert additional time and effort to convert progress into a solution.

Our first benchmark (`no time pressure') shuts down the first channel, and our second benchmark (`no payoff difference') shuts down the second channel.

\paragraph{No time pressure.} Facing an infinite time horizon, $T=\infty$, the thinking arm is a safe alternative for the agent. Thinking long enough guarantees progress and thus \emph{some} payoff. The following proposition describes our first benchmark. We relegate its proof along with all other proofs to the appendix. 

\begin{proposition}\label{prop:Tinf}
Suppose that the time horizon is infinite, $T=\infty$. Then, the agent either works first on the doing arm and eventually switches to the thinking arm or works on the thinking arm throughout.
The agent starts with doing if and only if \[\bar{p}\geq \hat{p} := \frac{c/\lambda}{B-V(\infty)+c/\mu}.\] She switches when $p_t=\hat{p}$, that is, at time $t=\overline{\tau}_1:=\max\left\{\frac{1}{\lambda} \ln\left(\frac{\bar{p}(1-\hat{p})}{\hat{p}(1-\bar{p})}\right),0\right\}$.
\end{proposition}

\Cref{prop:Tinf} shows that an agent who is sufficiently optimistic about the doing arm starts to work on it. An initial belief $\bar{p}\in(0,1)$ for such optimism exists if and only if the expected payoff from a good doing arm is higher than that from thinking, $B-c/\lambda > V(\infty)-c/\mu$. 

Two effects push the agent towards an initial doing period: lower expected cost and larger expected benefits. The first effect is present if a solution arrives faster than progress through thinking, $\lambda>\mu$. The second effect is present if a solution provides a higher payoff than progress even in the absence of time pressure $B>V(\infty)$. Higher payoffs from the doing arm may, for example, result from saving the additional effort required to convert progress.

\paragraph{No payoff difference.} In this benchmark, we assume that in the absence of time pressure, there is no payoff difference between (productive) arms; that is, we assume
\begin{align}\label{eq:benchmark_condition}
	\lim_{\tau \rightarrow \infty} V(\tau)&=B & c=0. \tag{C.1}
\end{align}
Condition \eqref{eq:benchmark_condition} ensures (i) that absent time pressure, successful thinking delivers no better or worse solution than doing, and (ii) that there is no difference in the expected cost of obtaining a solution and of obtaining and converting progress. This assumption allows us to focus exclusively on the role of time pressure. The following proposition describes the optimal policy in this case.

\begin{proposition}\label{prop:czero}
Suppose condition \eqref{eq:benchmark_condition} holds. Either the agent works first on the thinking arm and eventually switches to the doing arm, or she works on the doing arm throughout. The agent starts with thinking if and only if the deadline $T$ is large enough such that a solution $\tau_3 \in (0,T]$ to
\[\bar{p} = \frac{\mu V(\tau_3)}{B\left(\mu+(\lambda-\mu) e^{-\lambda \tau_3}\right)}\]
exists. In this case, the agent switches to doing when the time remaining $\tau_3$ is equal to the smallest of the solutions. \end{proposition}

The intuition behind \cref{prop:czero} is the following: If the deadline is close, even if progress arrives momentarily, then the time left to convert it is short. The payoff $V(\tau)$ vanishes fast in $\tau$. An arrival on the doing arm, instead, delivers a solution directly. As time runs out, the time pressure effect on the thinking arm trumps any fundamental uncertainty on the doing arm. The agent pulls the doing arm and throws a \emph{Hail Mary}.\footnote{ The term originates from American Football. In 1975, Dallas Cowboys quarterback Roger Staubach threw a 50-yard pass in the final seconds of a game, desperately hoping to make the game-winning touchdown. Staubach commented that while throwing the ball, he ``closed [his] eyes and said a Hail Mary''. Since then, throwing a Hail Mary has become synonymous with taking a risky action in desperation, often because time is nearly expired.}

\paragraph{Why and when to do?}\Cref{prop:czero} shows that the payoff motive in the absence of time pressure behind \cref{prop:Tinf} is not the only reason for doing: if time pressure is high, then gambling on risk, i.e., having a good arm, is more promising than gambling on time, i.e., managing to convert progress.  

Therefore, our benchmarks offer a first insight into why the agent opts for doing: (i) to materialize the payoff advantage of the doing arm\footnote{Recall that this payoff advantage can derive from effort-saving motives due to a faster arrival on the doing arm, from effort-saving motives due to not requiring an additional implementation stage, or from a higher payoff for a solution on the doing arm rather than completed conversion of progress.} and (ii) to succumb to time pressure. Our next step is to combine the two motives considering a setting with $T<\infty$ and $c>0$ for which the optimal allocation of effort is yet to be determined.

\subsection{Optimal Policy} \label{sub:optimal_policy}

We characterize the optimal policy in three steps. First, we state the agent's dynamic optimization problem. Second, we derive a set of necessary conditions for the optimal policy. These conditions have a straightforward economic interpretation that we discuss. Third, we state an algorithm which\textemdash under mild technical conditions\textemdash determines the uniquely optimal policy. The third step verifies the sufficiency of the necessary conditions.

\subsubsection{The Agent's Problem} \label{ssub:the_age}

By construction, the agent exerts full effort until the game ends. However, she dynamically decides whether to invest in thinking or in doing. Consider a situation in which the remaining time is $\tau$, and in which the agent holds a belief $p_\tau$ about the doing arm. Suppose the agent exerts effort $a_\tau$ on the doing arm for a small time interval of length $dt$. The instantaneous payoff of a solution is $B$. A solution arrives with probability $a_\tau p_\tau \lambda d t$. Suppose the agent exerts effort $1-a_\tau$ into the thinking arm for a small time interval of length $dt$. The instantaneous payoff of progress is $V(\tau)$. Progress arrives with probability $(1-a_\tau) \mu  dt$. 

The agent updates her belief about the doing arm according to Bayes' rule. Denote by $A_\tau:= \int_\tau^{T} a_s ds$ the amount of effort the agent has invested in the doing arm in the past. Then, the belief about the doing arm with time $\tau$ remaining and past effort $A_\tau$ on the doing arm is\[p_\tau = \frac{\bar{p} e^{-\lambda A_\tau}}{\bar{p} e^{-\lambda A_\tau}+(1-\bar{p})}.\]

If no arrival, i.e., neither a solution nor progress, occurs during the interval $[\tau, \tau-dt)$, then the payoff of progress declines to $V(\tau-dt)$\textemdash regardless of the agent's choice $a_\tau$. The belief, however, declines only when the doing arm was pulled with positive intensity, $a_\tau>0$.\footnote{In particular, the belief follows the standard ODE $dp_\tau/d\tau= p_\tau(1-p_\tau)a_\tau \lambda$, where, again, the notation follows the time remaining rather than the calendar time.}

The agent's objective is to dynamically maximize 

\[\max_{(a_\tau)_{\tau=0}^{T}}\int_0^T \underbrace{\vphantom{(}e^{-\mu(T-\tau - A_{\tau})}}_{P(\text{no progress yet})} \underbrace{(1-\bar{p} + \bar{p} e^{-\lambda A_{\tau}})}_{P(\text{no solution yet})} \underbrace{\left(\mu (1-a_\tau) V(\tau)+\lambda a_\tau p_\tau B \right)}_{\text{flow payoff}} d t\]
where the mapping $a_\tau:[0,T] \times [0,T-\tau] \rightarrow [0,1]$ determines the strategy with time $\tau$ remaining and past effort $A_\tau$ on the doing arm. We use $A_\tau$ as the state variable to derive the necessary conditions for an optimal strategy via optimal control methods. The formal details can be found in \cref{sec:notation}. We derive the following dynamic relative preference for the agent:

\[\gamma_\tau=e^{-\mu(T-\tau-A_\tau)} \left(\left(1-\bar{p}+\bar{p} e^{-\lambda A_\tau}\right)\mu V(\tau) - \bar{p} e^{-\lambda A_\tau}\lambda B \right)- {\eta}_\tau\]
where ${\eta}_\tau$ denotes the co-state of the optimal control problem. If ${\gamma}_\tau>0$, the agent pulls the thinking arm. If ${\gamma}_\tau<0$, the agent pulls the doing arm.

The co-state ${\eta}_\tau$ is determined by the boundary condition ${\eta}_0=0$ and its evolution
\begin{equation}\label{eq:evoeta}
	\begin{split} \frac{d\eta_{\tau}}{d\tau} = e^{-\mu(T- \tau-A_\tau)}\Bigg(& \mu(1-\bar{p}) \Big((1-a_\tau) \mu V(\tau)-c\Big) \\
	& - (\lambda-\mu) e^{- \lambda A_\tau} \bar{p} \Big((1-a_\tau)\mu V(\tau) + a_\tau \lambda B-c\Big)\Bigg).\end{split}
\end{equation}

\subsubsection{Necessary Conditions} \label{ssub:necessary_conditions}

We derive the necessary conditions for the optimal policy from Pontryagin's principle. These necessary conditions substantially reduce the space of the candidate strategies.

\begin{proposition}[Optimal Policy\textemdash Necessary Conditions]\label{lem:ded}
The optimal policy takes one of the following forms: \begin{enumerate}
\item the agent exclusively uses the doing arm,
\item the agent starts by thinking and switches to the doing arm, or
\item the agent begins with the doing arm, switches to the thinking arm eventually, and switches back to the doing arm when little time remains.
\end{enumerate}\end{proposition}

The critical insight leading to \cref{lem:ded} is that if the agent leaves the thinking arm once, she does not return to it. At a high level, the intuition behind this insight is the following. If the agent decides to leave the thinking arm, then only because she considers the value of progress to be too low due to the deadline approaching. Notably, the decline in the value of progress does not stop\textemdash even when pulling the doing arm.

While this observation is a substantial part of the story, it falls short in one aspect: Whenever the agent pulls the doing arm unsuccessfully, the payoff of that arm also declines because the belief about its state deteriorates.

The precise intuition behind the horse race of the two arms is subtle. A stepwise inspection of the effects at play is instructive. Consider the following equivalent formulation of the relative preference from above:

\[{\gamma}_\tau= {e^{-\mu(T-\tau-A_\tau)}}{\left(1-\bar{p}+\bar{p} e^{-\lambda A_\tau}\right)} \overbrace{\underbrace{\left(\mu V(\tau) - p_\tau \lambda B \right)}_{\substack{\text{payoff difference}}}- \underbrace{\eta_\tau}_{\substack{\text{effect of lower belief}\\ \text{on continuation value}}}}^{=:y_\tau}.\]

The agent thinks whenever ${\gamma}_\tau>0$. We focus on the second part, $y_\tau$, and consider an increase in the time remaining. Using the evolution of ${\eta}_\tau$ (\cref{eq:evoeta}), we obtain

\begin{align}\label{eq:gamme_evol}
	\frac{d y_\tau}{d \tau} = \underbrace{\mu V'(\tau)}_{\text{(i) deadline effect}} + \underbrace{\mu p_\tau \lambda ( V(\tau)-B)}_{\substack{\text{(ii) payoff-on-arrival}\\\text{effect}}} + \underbrace{(\mu-\lambda p_\tau) c}_{\substack{\text{(iii) effort-to-arrival}\\ \text{effect}}}.
\end{align}

Observe that \cref{eq:gamme_evol} is independent of the agent's action $a_\tau$; that is, the agent's action has no first-order effect on the evolution of the relative preference. Any effect that the agent's current action has on instantaneous payoffs is compensated by a dynamic effect in the continuation value\textemdash a feature common in the bandit literature.\footnote{In \Cref{sub:evolution_of_preference} we provide a derivation of \eqref{eq:gamme_evol} illustrating how the direct effect of the action drops out. Note, however, that there is a second-order effect of the agent's action through the belief $p_\tau$.}

As we increase the time to the deadline, \cref{eq:gamme_evol} describes the three incentives that determine the change in the agent's relative preference: (i) the change in the value of progress due to the reduced time pressure (the deadline effect); (ii) the payoff differential between the two arms upon an arrival (payoff-on-arrival effect); and (iii) the change in the difference in the expected effort required to reach progress or a solution (effort-to-arrival effect).

The deadline effect is the only effect that can be signed unambiguously and is always positive\textemdash pushing the agent towards the thinking arm. As the deadline moves further away, obtaining progress has a higher value. The additional time makes it more likely to convert progress before the deadline. The other two effects can be either positive or negative. The sign of the payoff-on-arrival effect depends on the relative payoffs between arms. If $V(\tau)<B$, the effect is negative\textemdash pushing the agent towards doing. As the deadline moves further away, a negative payoff-on-arrival effect pushes the agent to spend the additional time on the doing arm. The effort-to-arrival effect measures the relative expected cost difference of the arms. Thinking is expected to deliver an arrival faster than doing if $\mu>p_\tau \lambda$. In this case, there is a positive effort-to-arrival effect\textemdash pushing the agent towards thinking.

We derive the intuition for our critical insight---the agent pulls the thinking arm in at most one connected interval of time---from \cref{eq:gamme_evol} and the illustrated effects. We use the following construction. Suppose that the agent splits her thinking effort into two disjoint intervals of time. Then, there has to be a doing interval that is both preceded and succeeded by a thinking period. The existence of that interval implies that \eqref{eq:gamme_evol} is positive at the beginning of this doing period, i.e., when $\tau$ is high. In this case, an expanding time window pushes the agent toward thinking. At the same time, \eqref{eq:gamme_evol} is negative toward the end of the doing period, i.e., when $\tau$ is low. An expanding time window pushes the agent toward doing. 

To satisfy this property, \eqref{eq:gamme_evol} has to change signs during the doing interval from negative to positive as $\tau$ increases; thus, \eqref{eq:gamme_evol} has to cross zero from below. Formally, the agent's relative preference must attain an interior minimum during the doing interval. Economically, this implies that an increase in the time remaining---while the agent is pulling the doing arm---must change its effect on the relative preference: from pushing the agent further toward the doing arm to pushing her back toward the thinking arm. At such a minimum, the three effects must exactly balance each other. In particular, the positive deadline effect must be compensated for by the sum of the payoff-on-arrival and effort-to-arrival effects.

Observe that the deadline effect declines as the time remaining increases, $V''(\tau)<0$. Thus, time pressure becomes less of a consideration in the agent's decision. One force pulling the agent to the thinking arm becomes weaker as $\tau$ increases. Recall that the supposed strategy requires the relative preference to attain an interior minimum. For this to occur, the payoff-on-arrival and effort-to-arrival effects combined (i) must pull the agent sufficiently toward the doing arm to compensate for the deadline effect and (ii) must evolve sufficiently in favor of the thinking arm to dominate the decline in the deadline effect. However, both properties can never be satisfied simultaneously if the decline in the deadline effect is sufficiently strong. In particular, they cannot be satisfied if $V(\tau)$ is sufficiently concave: $-V''(\tau)\geq p_\tau \lambda V'(\tau)$. Thus, any doing period is preceded or succeeded by a thinking period but not both.\footnote{For the special case $\mu>\bar{p} \lambda$, observe that the effort-to-arrival effect is always positive. Moreover, the payoff-on-arrival effect has to be negative, $V(\tau)<B$, in any doing period that succeeds a thinking period. Thus, \eqref{eq:gamme_evol} increases in the time remaining only if $V'(\tau) + \lambda p_\tau V(\tau)$ increases. Given our concavity assumption, this is not the case. Hence, \eqref{eq:gamme_evol} is nonincreasing and $y_\tau$ concave.}

Intuitively, the agent will not think twice because of the following observation: If the agent, at some point, finds the switch \emph{doing $\rightarrow$ thinking} optimal, then she cannot have found the switch \emph{thinking $\rightarrow$ doing} optimal when she had more time remaining. To see this, note that at the switch \emph{doing $\rightarrow$ thinking}, the deadline effect must be weak because the agent found it optimal to do \emph{before} that switch.  However, as we move backward in time, i.e., increase the time remaining, the deadline effect weakens. If it weakens sufficiently fast, the agent is pulled more toward doing, the further away from the deadline. Thus, it is never optimal for the agent to think twice. The examples introduced in \Cref{ssub:examples} all satisfy this property.

\cref{lem:ded} therefore follows from these two observations: (i) the agent never returns to the thinking arm and (ii) a Hail Mary is inevitable once the agent runs out of time, which delivers \cref{lem:ded}. Either the agent starts with the Hail Mary or has at most one thinking period that precedes the Hail Mary period.

The necessary conditions provide substantial structure as they limit the set of possible solutions. In the following subsection, we derive an intuitive algorithm to compute the optimal policy based on the necessary conditions. Under a mild additional assumption, the algorithm delivers the unique solution to the agent's problem. Moreover, it provides further intuition for the economics of the agent's problem and comparative statics.

\subsection{Characterization of the Optimal Policy} \label{ssub:characterization}

The basic intuition for the optimal policy follows from combining the two benchmarks in the previous section: The agent may do when her belief about the doing arm is high, and the deadline is far as in \Cref{prop:Tinf}. Moreover, the agent will do\textemdash independent of her belief about the doing arm\textemdash when the deadline is close as in \Cref{prop:czero}.

We now construct a solution algorithm that, under the following assumption, delivers the unique solution to the agent's problem. Let
\[\hat{q}(\tau):=\frac{\mu \left(V(\tau)+c \tau\right)}{\mu (B+c\tau)+ (\lambda-\mu)\left(B-\left(1-e^{-\lambda \tau}\right)\left(B-\frac{c}{\lambda}\right)\right)}.\]

\begin{assumption}\label{ass:qincre} ~ \\
	(i) $V(\infty)\leq B+\frac{c}{\mu}$.\\ (ii) If $\mu>\lambda$, then $\frac{\mu V''(\tau)}{(\mu-\lambda)U''(\tau)}-\hat{q}(\tau)$ is monotonic.
\end{assumption}
\Cref{ass:qincre} is only a sufficient condition to ensure that a unique strategy satisfies the necessary conditions. It is a technical and by no means a necessary condition.\footnote{The crucial aspects for uniqueness are (i) that the belief at the beginning of the Hail Mary period is monotonic in the length of that period and (ii) that the maximum length of the thinking period is monotonic in the length of the Hail Mary period that follows. Both aspects must be true `in the relevant regions.' Our assumptions on primitives ensure that they are universally true. Moreover, uniqueness only facilitates the computation. If it fails, then our algorithm can be straightforwardly extended to determine all candidate solutions, which then have to be compared to determine the global solution. In particular, all solution candidates must satisfy \Cref{lem:ded}.} The first part of the assumption implies that the doing arm is sufficiently attractive to consider it a valuable arm beyond the Hail Mary period. 

The interpretation of the second part of the assumption is somewhat more subtle. Note first that condition (ii) is only relevant if progress on the thinking arm is expected to arrive faster than a solution on the doing arm conditional on the doing arm being good. In such a case, the condition ensures that $\hat{q}(\tau)$ is monotonic on the relevant part by requiring that the curvature of the arms' values is sufficiently regular.

Our algorithm constructs the optimal solution by working backward from the Hail Mary period. As we show in the appendix, condition (ii) guarantees that the length of the Hail Mary period is continuous and monotonic in the belief that the agent holds about the doing arm at the beginning of that period. This observation allows us to use marginal arguments to show uniqueness. Invoking \cref{lem:ded}, we state the optimal policy in terms of three variables 
\begin{enumerate}
	\item the time spent in the Hail Mary period, $\tau_3$,
	\item the time spent in the thinking period, $\tau_2$, and
	\item the time spent in the initial doing period, $\tau_1$.
\end{enumerate}
We are looking for a solution to the equation $\tau_1+\tau_2+\tau_3=T$. We make use of the following expressions:

\[q(\tau):= \min (1,\hat{q}(\tau)).\]

\[\dot{y}(s;p,\xi) := \left.\frac{d  y_{\xi+s}}{d s}\right|_{y_\xi=0}= \mu V'(\xi+s) + p \mu \lambda (V(s+\xi) - B) + (\mu- \lambda p)c, \text{ and }\]

\[\hat{y}(\tau;p,\xi):= \int_0^{\tau} e^{\mu s} \dot{y}(s;p,\xi) ds.\]

The first, $q(\tau)$, is the agent's belief when entering the Hail Mary period with time $\tau$ remaining. It originates from the agent's indifference between entering the Hail Mary period immediately and pulling the thinking arm for an infinitesimal measure of time before entering the Hail Mary period.

The second, $\dot{y}(s;p,\xi)$, describes the change in the relative preference due to a marginal increase in the deadline of an agent who enters a Hail Mary period with belief $p$ and time remaining $\xi$ and who pulls the thinking arm during the remaining times $[\xi+s,\xi)$ before switching to the doing arm with time remaining $\xi$. 

The third, $\hat{y}(\tau;p,\xi)$, describes the value of the relative preference with deadline $T=\tau+\xi$ by integrating over the survival-weighted evolution of the relative preference assuming indifference at the start of the Hail Mary period $\hat{y}(0;p,\xi)=0$.

We use these expressions to define the length of the initial doing period and the length of the thinking period as a function of the length of the final doing period: 
\[\tau_1(\tau_3):= \frac{1}{\lambda} \ln\left(\frac{\bar{p}}{1-\bar{p}}\frac{1-q(\tau_3)}{q(\tau_3)}\right)\text{, and}\]
\[\tau_2(\tau_3):= \begin{cases}
 \min \tau>0 \text{ s.t. }\hat{y}(\tau; q(\tau_3),\tau_3)=0, &\text{if a root for $y$ given $\tau_3$ exists,} \\ \infty &\text{otherwise.}
\end{cases}\] 
The first, $\tau_1(\tau_3)$, follows because the belief when entering the Hail Mary period is determined by the time spent doing in the initial doing period and Bayes' rule.

The second, $\tau_2(\tau_3)$, follows because indifference is necessary when switching from doing to thinking for the first time and when switching back.

We are now ready to state our algorithm that solves the fixed-point problem $T=\tau_1(\tau_3)+\tau_2(\tau_3)+\tau_3$ and thereby characterizes the agent's optimal strategy. We provide a further discussion of the algorithm after stating the characterization result.\footnote{A MATLAB program implementing the algorithm is available from the authors.}

\paragraph{Algorithm. }
 \begin{enumerate}
 	\item \label{step:initial} Set $\tau_1=\tau_2=\tau_3=0$.
 	\item \label{step:onlyd} Find the largest $\overline{\tau}_3$ such that 
 	\[\forall t \in [0,\overline{\tau}_3]\qquad q(\overline{\tau}_3-t)\leq\frac{\bar{p} e^{-\lambda t}}{(\bar{p} e^{-\lambda t}+1-\bar{p})}  .\]
 	 If $\overline{\tau}_3\geq T$, set $\tau_3=T$, $\tau_2=\tau_1=0$ and stop.
 	\item \label{step:canhaveed} If $q(\overline{\tau}_3)\neq \bar{p}$ go to \ref{step:dedHMmax}.
 	\item \label{step:ed} If $\tau_2(\overline{\tau}_3)\geq T-\overline{\tau}_3$, set $\tau_3=\overline{\tau}_3$ and $\tau_2=T-\overline{\tau}_3$ and stop.
 	\item \label{step:dedHMmax} Replace $\overline{\tau}_3$ by the largest $z$ such that
 	 \[ \forall t \in [0,z]\qquad q(z-t) \leq \frac{q(z) e^{-\lambda t}}{q(z) e^{-\lambda t}+ 1-q(z)} .\]
 	\item \label{step:convergence} Set $\tau_3=z$, $\tau_1=\tau_1(\tau_3)$ and $\tau_2=\tau_2(\tau_3)$. If $\tau_1(\tau_3)+\tau_2(\tau_3)+\tau_3=T$, stop. Otherwise, reduce $\overline{\tau}_3$ marginally and repeat \ref{step:convergence}.
\end{enumerate}

\begin{proposition}[Optimal Policy~\textendash~Characterization]\label{prop:characterisation}
Under \cref{ass:qincre}, the above algorithm determines the unique optimal policy.
\end{proposition}

\begin{figure}
\includegraphics[width=.43\textwidth]{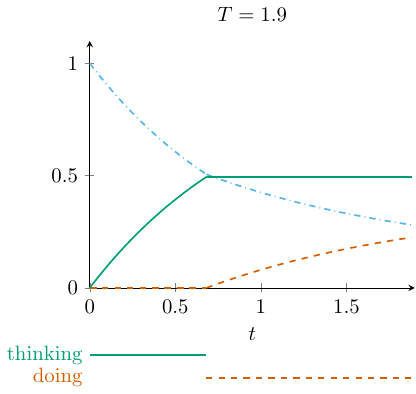}\hfill
\includegraphics[width=.43\textwidth]{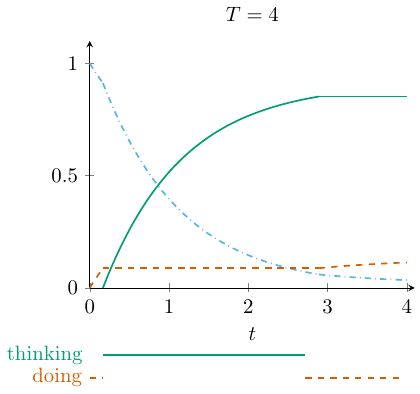}
\caption{\emph{Agent's optimal strategy and arrival probabilities.} The solid line plots the probability that the agent has made progress by time $t$, the dashed line the probability that the agent has found a solution through the doing arm by $t$, and the dash-dotted line the probability that the agent has neither made progress nor found a solution by $t$. Below we plot the time intervals in which the agent thinks or does absent any arrival. The left panel considers the optimal strategy given deadline $T=1.9$, and the right panel considers the optimal strategy given deadline $T=4$.\newline Parameters: $B=5,\bar{p}=3/4,c=1/2,\lambda=3/4,\mu=1$,$V(\tau)=(1-e^{-\tau})(B-c)$.}
\end{figure}

To build an intuition for the algorithm, recall that \cref{lem:ded} implies that it is without loss of generality to focus on three disjoint time intervals to characterize the solution: the Hail Mary period of length $\tau_3$, the thinking period of length $\tau_2$, and the initial doing period of length $\tau_1$. 

The algorithm constructs the solution via backward induction. It takes advantage of the property that the three time periods have to sum to the total time available to the deadline $T$\textemdash thereby constraining each other.

Given any length of the final doing period, $\tau_3$, we can determine the agent's belief upon entering the Hail Mary period. There are two cases: either (i) the agent enters the Hail Mary period immediately with belief $\bar{p}$, or (ii) she enters the Hail Mary period after at least one time interval of thinking. If, in addition, she enters the Hail Mary period with a belief $\hat{p}<\bar{p}$, then she must have pulled the doing arm before she started thinking.

First, in the case of (i), we need to ensure that the agent never finds it optimal to think for some positive measure of time until the deadline. To ensure this, the agent has to be sufficiently optimistic about the doing arm for any remaining time $\tau<T$. \Cref{step:onlyd} finds the largest deadline $T=\overline{\tau}_3$ such that the agent is sufficiently optimistic about pulling the doing arm throughout. It takes both declining beliefs and declining time windows into account. If $T\leq \overline{\tau}_3$, then the algorithm has found a solution.

Second, in the case of (ii), we need to ensure that the agent finds it optimal to switch to doing at the designated time $\tau_3$ and not to switch back to thinking thereafter. Thus, in addition to satisfying $q(\tau)<p_\tau$ for all $\tau<\tau_3$, we require $q(\tau_3)=p_{\tau_3}$. The latter ensures that the agent is indifferent between the arms when time $\tau_3$ remains. Optimality requires that the agent indeed prefers to think before the Hail Mary period. Using backward induction again\textemdash conditional on switching to doing at time $\tau_3$\textemdash the agent finds it optimal to think with remaining time $\tau+\tau_3$ if and only if the expression $\hat{y}(\tau';q(\tau_3),\tau_3)\geq0$ holds for all $\tau' \in [0,\tau]$. Conditional optimality follows because $\dot{y}(s;q(\tau_3),\tau_3)$ describes the evolution of the agent's relative preference between the arms\textemdash derived from the necessary conditions of the optimal control problem\textemdash assuming that she switches to the doing arm with time $\tau_3$ remaining. \Cref{step:ed} of the algorithm stops if $\hat{y}(\tau;q(\tau_3),\tau_3)\geq 0$ for all $\tau \in [0,T-\tau_3]$: the agent finds it optimal to start by thinking for a period of time $\tau_2=T-\tau_3$.

Third, if the agent engages in an initial doing period, we must be in case (ii). We know that the initial doing period determines the agent's belief for the final doing period via Bayes' rule as a function of the length of the initial doing period, $\tau_1$, and the agent's ex ante belief, $\bar{p}$. The expression $\tau_1(\tau_3)$ describes the time that the agent has to experiment without success on the doing arm such that her belief deteriorates to $q(\tau_3)$.

Fourth, the agent needs to be indifferent both after the initial doing period and when starting the final doing period. Whenever the agent finds it optimal to spend time $\tau_2=T-\tau_1(\tau_3)-\tau_3$ in the thinking period, $\hat{y}(\tau_2;q(\tau_3),\tau_3)=0$. If the expression $\tau_2(\tau_3)>0$, then that indifference is guaranteed with time remaining $\tau_2+\tau_3$. If, instead, $\tau_2(\tau_3)=\infty$, then it is never optimal to leave the initial doing period with a belief $q(\tau_3)$. \Cref{step:convergence} of the algorithm stops only if all conditions are met and thus determines the fixed point $T=\tau_1(\tau_3)+\tau_2(\tau_3)+\tau_3$.

\section{Application: Entrepreneurial Problem Solving} \label{sec:finalremarks}

Our model highlights the tradeoff that an agent faces under time pressure: should she try to apply an uncertain method ready at hand, or should she take a step back and develop a different method that involves less fundamental uncertainty? The optimal strategy is a function of both the time horizon and her belief about the initial method. Our results emphasize how learning and time pressure interact. 

We now apply our findings to our motivating application\textemdash an entrepreneur's decision to meet a target to secure follow-up financing by a deadline. We first relate the model to the specific context of entrepreneurs who need to achieve a milestone to obtain the next round of funding. After that, we derive implications from our main theoretical results regarding the application considered and relate it to empirical phenomena. 

\subsection{Entrepreneurial Problem Solving} \label{sub:entre}
Consider an entrepreneur who has raised funding for her venture. She has to prove the business's prospects by some deadline. Deadlines are ubiquitous in innovative entrepreneurship. Among many other reasons, they may come from (i) funders explicitly setting deadlines \citep[e.g., via staged contracts, see][]{kaplan2003}, (ii) the need to raise new funds before the startup runs out of cash,\footnote{See \citet{whystartupsfail}\textemdash running out of cash is the most frequent reason startups fail based on CB Insights' analysis of startup failure post-mortems.} or (iii) implicitly, according to the expectation that the market moves on after some time either by changing focus or by adopting a competitor's product.\footnote{The drone analytics provider, Airware, went out of business because, initially, they bundled their software with a self-engineered drone. However, once they were ready to launch the bundled product, cheaper alternative drones were already available. They pivoted to focusing on software development only but ran out of money and eventually ceased operations. Perhaps more famously, despite being a corporate favorite, Blackberry failed to innovate until it missed its deadline and the market had moved to iOS and Android. See also the discussion in \citet{gans2018strategy}.}

The flow cost of working on the problem, $c$, has two interpretations in the context of startups. The first interpretation is literal and derives from what is called the startup's \emph{burn rate}. Each period, the startup has to pay its employees, rent an office or lab space, purchase equipment, etc. If the startup has access to initial funds $C$, then the burn rate implies an implicit deadline $T=C/c$ by which it has to have raised new funds.\footnote{The fact that we use continuous time with a bounded per-unit effort strengthens this interpretation. Instead of seeing the deadline as a clock ticking, we could interpret $T$ as the total effort budget available. Investing $c dt$ units of effort on doing or thinking from this budget implies an arrival with rates $p_\tau \lambda dt$ and $\mu dt$, respectively. The value of progress depends on the effort remaining within the budget. Having invested $T$ without a solution makes the agent perish.} The entrepreneur wants to complete the task with funds remaining in her pocket under this interpretation. She can invest these leftover funds in later stages. The second interpretation is to consider $c$ as the agent's (linear) time cost. Such an interpretation is proposed by \citet{whystartupsfail2}. He argues that entrepreneurs have a direct time cost in the form of an action bias and prefer to get things done as quickly as possible.

Our focus is on the entrepreneur's approach in trying to meet the requirements before a deadline, for example, successfully launching a product, developing a new product, improving an existing product, or meeting a revenue threshold. We assume that the entrepreneur has an initial idea that she is not fully sure is suitable for completing the next step. The entrepreneur can try to go to the next step without further ado\textemdash she pulls the doing arm. Alternatively, she can attempt to pivot. To prepare the pivot, she searches for a better approach to meet the target\textemdash she pulls the thinking arm. Investing time and effort into a change of the startup's strategy is commonly observed \citep[see, for example,][]{Kirtley2021} but has received little theoretical attention. 

Our model of the doing arm resembles the standard experimentation approach. It is commonly used in modeling entrepreneurial strategy \citep[for an overview, see][]{kerr2014}. The thinking arm captures that the value of new ideas depends on the resources and time available to convert them into solutions. To make the difference between fundamental risk and time risk clear, \cref{example_implications}, assumes no fundamental uncertainty about the thinking arm. However, that assumption is not crucial and the other examples introduced in \Cref{ssub:examples} share the same qualitative features.

\subsection{Implications} \label{sub:implications}

Both beliefs and the time horizon matter for entrepreneurs when contemplating how to invest resources in their venture \citep[see, for example,][]{Kirtley2021,Rahmani2021}. Moreover, entrepreneurs have a tendency to \emph{do} early to ``get things done'' \citep[see, for example,][]{whystartupsfail2,gans2018strategy}. Translated to our model, entrepreneurs take risks early on to arrive at a solution quickly. This observation is in line with our finding that the agent starts by doing if she is sufficiently optimistic about her initial approach and the time pressure is not too high initially (see \cref{lem:ded}).

However, this strategy comes at a cost: The entrepreneur reduces the expected time to solve the current problem by doing early. At the same time, this strategy reduces the probability of solving the problem in time. The reason is that hoping for an early solution produces a \emph{false start} \citep{whystartupsfail2}: Doing early delays thinking about a pivot\textemdash e.g., shifting business to accommodate a different market\textemdash yet pivots occur with positive probability. If the entrepreneur ends up pivoting, she suffers from the shrunk time window.

We put structure on the thinking arm to address false starts formally. We wish to compare how the agent trades off the expected effort cost against the probability of finding a solution. Therefore, we need to take a stance on \emph{how} the agent converts progress into a solution. For clarity, here, we restrict attention to \Cref{example_implications}: An arrival on the thinking arm delivers a new bandit with a known arrival rate $\nu \geq \lambda \overline{p}$. To simplify further, we assume $B_\nu=B$ and $c_\nu=c$. The properties of \cref{example_implications} resemble the discussion in \citet{whystartupsfail2}, in particular, the discussion of the Triangulate venture and the choices of its founder.

As a first result, we see that thinking early and backloading doing improves the ex ante probability of obtaining a solution in \cref{example_implications}.
\begin{proposition}\label{prop:doing_hurts_success}
	Consider \cref{example_implications} with $B_\nu{=}B,c_\nu{=}c$. For any potential strategy inducing $\tau_1> 0, \tau_2>0$, and $\tau_3>0$, backloading all effort on the doing arm, i.e., choosing $\tau_1'=0,\tau_2'>\tau_2$, and $\tau_3'=\tau_1+\tau_3$, increases the probability that a solution is found by the deadline.
\end{proposition}

\cref{prop:doing_hurts_success} gives a theoretical foundation for the empirical phenomenon of \emph{false starts}\textemdash the delay or entire absence of customer research before launching a minimum viable product. Whenever the entrepreneur starts with her initial idea right away, i.e., does early on, she sacrifices success probability\textemdash unless the initial deadline is very short. To see why this occurs, observe that it is straightforward to rewrite the agent's problem as
\begin{align*}
	\max_{\textbf{a}:=(a_t)_{t=0}^T} \mathbb{P}^{\textbf{a}}\left[\text{solution before }T\right]B - \mathbb{E}^{\textbf{a}}\left[\text{time worked}\right]c.
\end{align*}
This rewriting makes it apparent that the agent wants to balance the probability of success against the expected time to solve the problem.

\begin{figure}
\includegraphics[width=.45\textwidth]{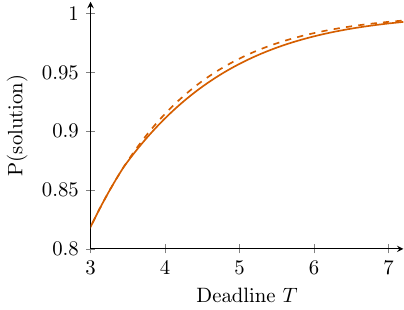}\hfill
\includegraphics[width=.45\textwidth]{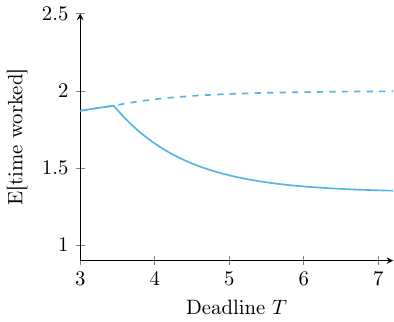}
\caption{\emph{Solution probability vs. cost reduction.} The left panel plots the ex ante probability of finding a solution against the ex ante deadline length. \newline The right panel plots the expected time the agent works. \newline Solid lines represent these under the agent-optimal strategy, and dashed lines represent these using the same total amounts of doing and thinking but backloading all doing. For deadlines shorter than the depicted range, the agent enters the Hail Mary period immediately, and the two curves coincide.\newline  Parameters: $B=5,\bar{p}=3/4,c=1/2,\lambda=3/4,\mu=\nu=1$.}\label{fig:SolProbCost}
\end{figure}

Within our model, the agent aims to reduce the expected effort because it is costly. In reality, there are multiple underpinnings for such effort cost: among these is a direct disutility of effort, the desire to save funds for the future, or a bias for moving forward fast with the venture. \Cref{fig:SolProbCost} highlights the consequences. The larger the time horizon is, the more the agent saves on her expected effort. Perhaps surprisingly, the time the agent expects to work can decline in the deadline length. The reason is straightforward: the agent adjusts her strategy to do early. If doing is successful, then she finishes earlier, which, in turn, reduces the expected effort invested. However, the agent's investment choices come at the cost of reducing the expected probability of succeeding at all.

From a venture capitalist's perspective, the return of marginally expanding the entrepreneur's deadline may thus not fully translate into an increase in the probability of finding a solution\textemdash even though we expect the agent to work absent a solution. Instead, the entrepreneur may sacrifice some of the extra potential to arrive at a solution faster. 

\begin{proposition}\label{lem:tau_cs}
		The length of the initial doing period, $\tau_1$, and the length of the thinking period, $\tau_2$, are nondecreasing in $T$. As $T \rightarrow \infty$, $\tau_2 \rightarrow \infty$. 
\end{proposition}

\begin{figure}
\includegraphics[width=.46\textwidth]{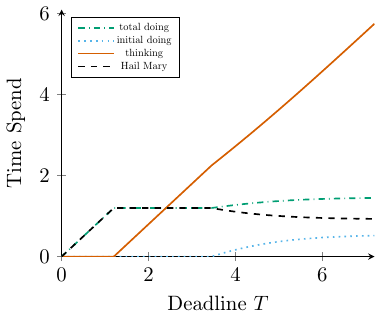}\hfill
\includegraphics[width=.46\textwidth]{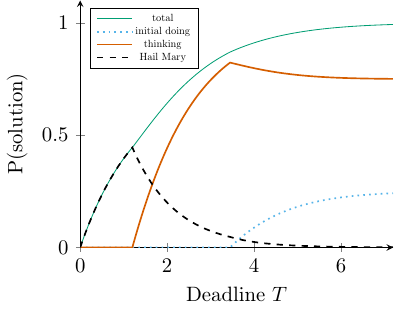}\hfill
\caption{\emph{Time spent in different periods (left), and probability of a solution by period (right) for different ex ante deadlines $T$.} The left panel plots the time spent in different phases under the optimal strategy against the initial time horizon: initial doing (dotted), thinking (solid), and Hail Mary (dashed). The dash-dotted line depicts the maximum time spent on the doing arm.\newline
The right panel plots the ex ante probabilities of obtaining a solution by phase against the initial time horizon: initial doing (dotted), Hail Mary (dashed), and thinking\textemdash progress \& conversion (thick solid). The thin solid line is the aggregate probability that a solution occurs (the sum of the other three curves).\newline
\emph{Note:} This figure compares different ex ante time horizons and must not be confused with the agent's decision over time. Parameters: $B=5,\bar{p}=3/4,c=1/2,\lambda=3/4,\mu=\nu=1$.}\label{fig:periodsandprobs}
\end{figure}

\cref{lem:tau_cs} shows that in the beginning, the agent never decreases the time devoted to doing. As the left panel of \cref{fig:periodsandprobs} suggests, the frontloading of doing and thus false starts become a larger problem with longer deadlines. Therefore, an increase in the deadline can not improve the induction of false starts. However, once the deadline offered is long enough, the problem becomes second-order: $\tau_2$ becomes arbitrarily large, and the probability of obtaining some solution converges to 1.

Depending on the product at hand, the venture capitalist may not only care about the entrepreneur finding a solution but also benefit from potential externalities depending on the solution method. For example, venture capitalists may benefit from customer research for other projects. In particular, if\textemdash as \citet{gompers1995,kaplan2003} suggests\textemdash the entrepreneur's motivation within a financing stage comes mainly from meeting the explicit requirements, i.e., from solving the problem in time, the venture capitalist's primary instrument is to expand or tighten the duration of the stage.

As we see in the right panel of \cref{fig:periodsandprobs}, tightening the entrepreneur's deadline may increase the probability that she finds her solution using customer research. The reason is that with low time pressure, the entrepreneur spends large portions of the extra time gambling on a quick and successful launch of her initial minimum viable product. Even if she fails initially, she remains confident that there is enough time to generate and convert insights from customer research. With a tighter deadline, the same entrepreneur engages in customer research earlier.

A direct consequence of this observation is that a venture capitalist may provide deadlines that are shorter than his actual time horizon to induce\textemdash perhaps surprisingly\textemdash a more thorough approach by the entrepreneur. The deadline discourages the entrepreneur from gambling on quick successes and incentivizes her to begin the project with customer research.\footnote{Interestingly, external risk, described as risks equally uncertain to both the venture capitalist and the entrepreneur\textemdash e.g., future demand for an undeveloped product\textemdash significantly lowers the time until the subsequent financing round in \citet{kaplan2003}. Our results provide one mechanism that can rationalize this observation: when external risk is high, the value of customer research is high. Venture capitalists can encourage early customer research with intermediate deadlines: long enough to prevent an immediate Hail Mary but short enough to discourage a false start.}

Unsurprisingly, the agent's time spent on each approach is a function of the agent's initial belief. For example, suppose the entrepreneur is pessimistic that a launch will succeed without additional insights from customer research. In this case, she is unwilling to launch it\textemdash unless the time window is small\textemdash and rather engages in customer research first. Instead, if the entrepreneur is optimistic, she tries launching first to save her effort on customer research.

Indeed, as the following proposition shows, if the initial belief $\bar{p}$ is large, then---independent of the deadline---the agent never begins by thinking. At the same time, if the initial belief is low, then the agent only starts with doing when under immediate time pressure.

Recall $\hat{p}$ from \cref{prop:Tinf},

\[\hat{p} = \frac{c/\lambda}{B-(V(\infty)-c/\mu)}= \frac{\mu \nu}{\lambda (\mu+\nu)}\]
which, in \cref{example_implications}, is independent of $B$ and $c$.\footnote{The proof of \cref{prop:beliefs} also makes the proposition applicable outside of \cref{example_implications}.}

\begin{proposition}\label{prop:beliefs}
Consider \cref{example_implications} with $B_\nu{=}B,c_\nu{=}c$. Fix $B,c,\lambda,\mu$, and $V$ such that \cref{ass:qincre} holds. Independent of the time horizon $T$, the following statements hold:\begin{enumerate}
	\item \label{it:overlinep} there is a $\tilde{p}$ such that if $\bar{p
}>\tilde{p}$, then the agent begins with a doing period and $\tilde{p}$ solves
\[ \tilde{p}=\frac{V'(q^{-1}(\tilde{p}))+c}{\lambda (B+c/\mu - V(q^{-1}(\tilde{p}))};
\]
	\item \label{it:hatp} if $\bar{p
}< \hat{p}$, then the agent switches arms at most once and only from thinking to doing;
	\item \label{it:underlinep} if $\bar{p
}\geq \hat{p}$, then the agent's belief never falls below \[\check{p}:=\frac{\hat{p} e^{-\lambda q^{-1}(\hat{p})}}{\hat{p} e^{-\lambda q^{-1}(\hat{p})}+ 1-\hat{p}}\geq \frac{\bar{p
} e^{-\lambda T}}{\bar{p
} e^{-\lambda T}+ 1-\bar{p
}}.\]
\end{enumerate} 
\end{proposition}

\Cref{prop:beliefs} provides insights into the potential for the venture capitalist who receives a payoff $\Pi>0$ if the entrepreneur successfully launches the product in some way before a deadline $T^{VC}$. 

Note that the venture capitalist's first best is identical to that derived in \cref{prop:czero}. The venture capitalist does not incur the entrepreneur's cost. Whether he can implement his first best depends on whether he can design a contract $(B,T)$ such that the agent switches once and at the right time.

It is trivial that the venture capitalist can induce his first best if his preferred strategy is to throw the Hail Mary throughout, i. e., if his own deadline $T^{VC}$ is relatively short. However, for larger deadlines, he has to find a payment $B$ such that the entrepreneur's optimal switching time $q^{-1}(\bar{p})$ coincides with his optimal switching time $\tau^{VC}$. It turns out that such a payment may not exist. For example, for the case $\mu\geq \lambda=\nu=1$, no such $B$ exists for any $\bar{p}, B>0$, and $c>0$.\footnote{The venture capitalist's switching time, $\tau$, solves $\bar{p}=\frac{\mu V^{VC}(\tau)}{\Pi\left(\mu + (\lambda-\mu)e^{-\nu \tau}\right)}= \frac{\mu (1-e^{-\nu \tau})}{\left(\mu + (\lambda-\mu)e^{-\lambda \tau}\right)}$;\\ the entrepreneur's switching time, $\tau$, solves  $\bar{p}= \frac{\mu \left((1-e^{-\nu \tau})(B-c/\nu)+ c \tau\right)}{\mu\left(B + c \tau \right) + (\lambda-\mu)\left(B-(1-e^{-\lambda \tau})(B-c/\lambda)\right)}$. If $c>0$, for example, then both equations cannot hold for any $B$ if $\mu\geq\nu=\lambda=1$.}

If the venture capitalist cannot control $B$, for example, because the entrepreneur is motivated by success rather than by payments from the venture capitalist and does not face immense time pressure, then obtaining the venture capitalist's first best reduces at most to a nongeneric coincidence. Whenever the belief about the entrepreneur's initial idea is too high, $\bar{p}>\tilde{p}$, then achieving it within any deadline is impossible. 

\cref{prop:beliefs} admits the following corollary, which has further implications on how venture capitalists can use time pressure to induce entrepreneurs to exert desired actions.
 
\begin{corollary}\label{cor:maxdoing}
The agent's belief during the thinking phase is larger than $\min\{\hat{p},\bar{p}\}$. If $\bar{p}\leq \hat{p}$, then the probability that the agent solves the problem through the doing arm is maximized with deadline $T_1=q^{-1}(\bar{p})$.
\end{corollary}

The corollary states that the agent is constrained in the amount of experimentation she is willing to exert by her option to think instead.

The corollary is, for example, relevant in the following setting. Suppose that doing corresponds to the entrepreneur trying to launch a product in a particular business-to-business (B2B) context. In contrast, thinking corresponds to exploring potential direct-to-consumer (D2C) markets where the entrepreneur can pivot with her product. Suppose further that it is known that B2B is not the ideal market so the feasibility of a successful launch is uncertain. However, it is also uncertain \emph{which} D2C market is the right market. To resolve this uncertainty, the entrepreneur needs to carry out customer research. However, the venture capitalist may be interested in entering the B2B market to establish his reputation and therefore may have the preference (subject to success) that the entrepreneur launches in the B2B market.

Because the entrepreneur can pivot, she is not exploring all the options to launch in the B2B market. By imposing time pressure on the entrepreneur, the venture capitalist can maximize the chances of entering the B2B market.\footnote{Although the possibility result stated in \cref{cor:maxdoing} relies on the fact that a pessimistic entrepreneur is not going to have an initial doing phase, it is often the case that even if an infinite deadline involves an initial doing phase, the likelihood of obtaining a solution through the doing arm is maximized with $T=T_1$; see, e.g., the configuration in the left panel of \cref{fig:periodsandprobs}: Even as $T \rightarrow \infty$, the dash-dotted line will not be higher than at the interior maximum.}

\begin{figure}\centering
\includegraphics[width=.43\textwidth]{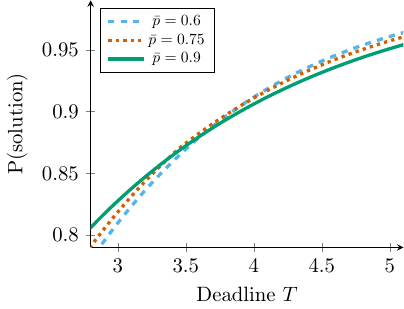}
\caption{\emph{Solution probability for different initial beliefs $\bar{p}$ against the initial time horizon.}\newline Parameters: $B=5,\bar{p}=3/4,c=1/2,\lambda=3/4,\mu=\nu=1$.}\label{fig:nonmonop}
\end{figure}

Although one may conjecture that higher initial beliefs $\bar{p}$ always increase the ex ante probability of obtaining a solution, this need not be the case. A higher initial belief may push the entrepreneur toward increasing her initial efforts to obtain a solution through the doing arm and to save on the cost of effort. Such a choice, in turn, may lower the probability of obtaining a solution in the given time frame. See \cref{fig:nonmonop} for an example.\footnote{We cannot derive meaningful conditions for when that nonmonotonicity occurs, but numerically it appears relatively robust.}

We want to stress that our findings in this section are in stark contrast to those from a classical infinite horizon two-armed bandit model in which both arms have time-independent payoffs, and one arm has a lower arrival rate. False starts would not arise in the canonical model. There, for any given sequence of actions, changing their order does not influence the overall success probability. Moreover, in the canonical model, the agent pulls the risky arm until the arms' instantaneous arrival rates are equal. However, the false-start notion builds on the idea that the agent conducts customer research too late. The canonical model unambiguously predicts nondecreasing probabilities of solving the task with a particular arm for increases in both the deadline and the initial belief. In contrast, our model highlights a significant economic incentive that is absent in a model that does not feature a time dependence of the thinking arm's payoff: the deadline regulates not only the overall success probability but also how success is achieved.\footnote{Indeed, it may also happen that an increase in the deadline reduces the probability of a solution through the doing arm\textemdash see, e.g., the left panel of \cref{fig:periodsandprobs}. Thus, an increase in the deadline\textemdash holding the initial belief about the doing arm fixed\textemdash increases the probability of a \emph{pivot}.}

\section{Discussion of Modeling Choices} \label{sec:discussion}
In this section, we discuss the motivation behind our modeling choices and their direct consequences on a more abstract level. We specifically want to emphasize the role of three model ingredients: (i) the model includes a deadline but excludes exponential discounting; (ii) the value of successful thinking diminishes and does so at a sufficiently increasing rate; and (iii) the value of successful thinking is independent of the belief about the doing arm.

\paragraph{Time Cost.} Our choice not to include standard exponential discounting is motivated by our focus on the changing time pressure. The implicit assumption in (infinite-horizon) exponential discounting models is that time pressure is constant at any point in time\textemdash e.g., because the risk of an exogenous termination of the game is constant. However, with a deadline in mind, this form of time pressure becomes less relevant\textemdash at least close to the deadline. Instead, the agent fears that she has insufficient time to finish her task before the game ends with certainty.

In our applications, deadlines are foreseeable dates on the time horizon, and time pressure increases as agents move closer to the deadline. Therefore, to ensure a transparent and tractable discussion of the effect of changing time pressure, we abstract from additional exponential discounting.

However, it should become clear from the analysis that including exponential discounting would not alter the economic effects of our model but would come at a substantive loss of tractability.\footnote{Indeed, it is straightforward yet cumbersome to adjust our key lemmata to include exponential discounting and to verify that \Cref{lem:ded} continues to hold. The same holds when deadlines arrive stochastically but become more and more likely as time progresses.} We focus on a world in which the agent has no incentive to shirk. As we see, e.g., in \cref{fig:SolProbCost}, the agent has an incentive to obtain results early, even absent exponential discounting: the agent incurs the cost of experimentation.

\paragraph{Diminishing Value of Progress.} We assume that the returns to thinking diminish at an increasing rate as the deadline approaches. This assumption captures the idea that successful thinking implies progress but not a solution. When progress arrives, the closer the deadline is, the less time remains to convert the progress made\textemdash the return shrinks. 

In light of our application, the simplest interpretation of this assumption is to think of progress as triggering a random process determining ex post payoffs. Thus, the greater the time remaining is, the more likely the agent can succeed in time despite a sequence of adverse shocks. That notion of progress differs from one in which a solution arrives deterministically with some delay. In the latter world, the value of progress is a positive constant until it drops to zero once the remaining time falls below a threshold.

In our model, \emph{any} decrease in the time remaining upon the arrival of progress reduces its value: there is less time to convert progress into a solution. Hence, the agent faces a crucial time tradeoff: delaying thinking reduces the expected time remaining when progress arrives and thus makes it less valuable.

\paragraph{Relative Concavity of the Value of Progress.} The primary assumption that leads to \Cref{lem:ded} is that the relative concavity of the value of progress is sufficiently high. This assumption implies that the evolution of the deadline effect dominates the other effects on the agent's preference as the time remaining increases. This assumption helps to focus on the main tradeoff between risk and time pressure. However, we could weaken the assumption without losing tractability. The main difference is that we may obtain two disjoint intervals in which the agent pulls the thinking arm and two disjoint intervals in which the agent pulls the doing arm. The underlying economic reason is that once the time pressure has become relatively weak, the thinking arm may have a payoff advantage over the doing arm.

\begin{example}\label{par:ex_double_think} Consider a variant of \cref{example_implications} in which $\nu<\bar{p}\lambda$; that is, the arrival of progress triggers a new bandit arm that has a relatively low arrival rate. It follows that the change in the deadline effect, $V''(\tau)$, will be dominated by the change in the payoff-on-arrival effect, $p_\tau \lambda V'(\tau)$, when $p_\tau$\textemdash and thus the weight on the payoff on the arrival effect\textemdash is relatively high. \end{example}

In particular, to further simplify this example for illustrative purposes, assume that $B_\nu=B+\frac{c}{\mu}$ and $c_\nu=0$. Once the belief has sufficiently deteriorated such that $p_\tau \lambda = \nu$, the continuation game is nested by our model. \Cref{lem:ded} applies. We show numerically that the optimal policy may have the following structure: the agent (i) starts thinking, (ii) switches to doing, (iii) switches back to thinking, and (iv) and returns to the doing arm for a Hail Mary. For an intuition of the resulting changes and additional details, see \Cref{sub:discussion_of_example_par:ex_double_think}.

\paragraph{Value of Progress Independent of the Belief about the Doing Arm.} We assume that the belief about the doing arm has no direct consequences on the payoff of the thinking arm. With time $\tau$ remaining, the agent attaches the same value to the thinking arm if she is almost certain that the state of the world is $\theta=0$ or if she is almost certain that it is $\theta=1$. While we show a set of examples in \cref{ssub:examples} that satisfy this assumption, it is nevertheless restrictive. We impose it, as it tremendously simplifies the analysis. 

However, this assumption is not crucial for our results. The following example\textemdash a version of \cref{par:ex_risky}\textemdash provides a setting in which the value of progress depends on the belief, yet our main result remains unchanged.

\begin{example}\label{par:ex_riskyreturn} As in \cref{par:ex_risky}, progress implies a new risky arm with intensity $\nu$, but suppose now that switching back to the initial doing arm is costless. Suppose further that $\nu=\lambda=1$, $\bar{p}\leq 2/3$, and $B>2c$. In this case, the entrepreneur will split her effort equally between the two arms once $p^\nu_\tau=p_\tau$.\footnote{Note that upon arrival of a new arm with $\overline{p}^\nu>p_\tau$, the entrepreneur will first only pull the new arm. Thus, its belief will decline while the other arm's belief remains constant until the beliefs on both arms are equal.} \end{example}

In \cref{sub:belief_V} we provide the formal analysis for a model in which the value of progress depends on the agent's belief about the doing arm that nests \cref{par:ex_riskyreturn}. In addition, we derive and discuss a condition that ensures that \cref{lem:ded} remains valid.

The only change in the generalization in \cref{sub:belief_V} is that we formulate the value of progress as a function of both the time remaining and the belief that the agent holds at the time of progress. The condition we derive is directly on this value of progress. Therefore, the model outlined in \cref{sub:belief_V} also captures other extensions to the baseline case. One such extension is a world in which the value of progress correlates directly with the underlying state $\theta$. Such correlation is relevant, for example, if the absence of success is informative about the \emph{problem's} difficulty rather than about the quality of the doing arm.

\section{Final Remarks} \label{sec:conclusion}

We address a time-constrained agent's dynamic decision when to \emph{do}\textemdash address a problem using an initial idea\textemdash and when to \emph{think} about an alternative, less risky method. We show that she should think neither too early nor too late. Overall, the agent never thinks twice. Once she stops thinking and moves to her initial idea, she abandons thinking for good.

We can use the specification of our model from \cref{sec:finalremarks} to predict the time an agent needs to find a solution as a function of her initial deadline. Depending on the expected speed of the various arms, perhaps counterintuitively, we may see that the average time an agent needs to find a solution decreases in her initial deadline. Such a prediction is\textemdash in some settings\textemdash directly testable. For example, suppose that the thinking and conversion process is sufficiently fast in a field of research. In this case, we would expect that researchers who undergo evaluation after an intermediate tenure clock would take on average \emph{longer} to fulfill the tenure requirements than those with longer or shorter tenure clocks. The mechanism is the following: those with shorter tenure clocks fail more often: the late bloomers drop out and the average time conditional on making tenure is shorter. Those with longer tenure clocks work first on converting their job market papers to influential publications and begin to branch only upon failure. Since some succeed, they need considerably less time to fulfill the requirements\textemdash they outpace their counterparts with intermediate clocks who branch from the beginning.\footnote{We want to emphasize that this result does not hold under arbitrary arrival rates. Therefore, to test these predictions, one must understand the arrival rates in place.}

In addition, our result from \cref{prop:doing_hurts_success} can serve as a cautionary tale regarding the efforts by politicians to meet specific pre-committed goals. For example, consider the Paris agreement. Many countries aim to target their goals through small-scale efforts such as incentivizing the use of electric cars. However, potentially, there is an action bias. Governments prefer implementing measures that are directly at hand\textemdash even at the risk of these being insufficient\textemdash instead of investing in ``transformational change'' that is warranted, e.g., by \citet{IPCC15}. Only if it becomes imminent that small-scale effort will not be sufficient will governments pivot to thinking about transformation. The late pivot is optimal from the government's perspective. Nevertheless, it occurs too late on a larger scale: the chances of meeting the goals in time decline. 

 \newpage
\appendix \setstretch{1}
\renewcommand{\partname}{}
\renewcommand{\thepart}{}
\part{Appendix}\setlength\bibitemsep{.8\itemsep}

\section{The Optimal Control Problem} \label{sec:notation}

\subsection{Notation} \label{sub:notation}

\renewcommand\tabularxcolumn[1]{m{#1}}
\renewcommand{\arraystretch}{1.1}
\begin{table}[h!]\begin{center}
\begin{tabularx}{.9\textwidth}{cX}
Variable & Description \\ \toprule
$\bar{p}$& Ex ante probability that $\theta=1$.\\ 
$B$& Benefit of arrival on the doing arm.\\
$V(\tau)$& Value of progress with time $\tau$ remaining. \\ 
$c$& Flow cost of effort. \\ \midrule
$\tau:=T-t$& Time until the deadline.\\ 
$a_\tau$ & Relative intensity of pulling the doing arm with $\tau$ periods remaining.\\ 
$A_\tau:=\int^{T}_\tau a_s ds$ & Total amount of time spent pulling the doing arm in the past. \emph{The state.}\\ \midrule
$p_\tau:=\frac{\bar{p}e^{-\lambda A_\tau}}{\bar{p} e^{-\lambda A_\tau} + 1-\bar{p}}$& Likelihood that the doing arm is suitable with $\tau$ periods remaining.\\
$U(\tau):=(B-c/\lambda)(1-e^{-\lambda \tau})$& Value of pulling an arm with known intensity $\lambda$ for time $\tau$.\\  
$\gamma_\tau:= -\frac{d H_\tau(a_\tau;A_\tau)}{d a_\tau}$& Relative preference for thinking with $\tau$ periods remaining.\\ 
$\eta_\tau$& Co-state with $\tau$ periods remaining.\\ \midrule
$Z^d(p,\tau):= pU(\tau)- (1-p)c\tau$ & Expected value from pulling the doing arm throughout with time $\tau$ remaining when the belief is $p$. \\
{ $\begin{aligned}
Z^t( \varepsilon; p,\tau):=&\left(\mu V(\tau) - c\right) \varepsilon\\&{+} (1{-} \mu \varepsilon) Z^d(p,\tau{-}\varepsilon)\\ &{+}o(\varepsilon)
\end{aligned}$} & Expected value from pulling the thinking arm for a small measure of time $\varepsilon$ and pulling the doing arm for the remaining time with time $\tau$ remaining when the belief is $p$.\\
\bottomrule
\end{tabularx}\end{center}
\end{table}

\subsection{Necessary Conditions for Optimality} \label{sub:hamilt}
Much of our arguments rely on the necessary conditions from Pontryagin's maximum principle. The existence of an optimal control follows from standard arguments.\footnote{See, for example, \citet{clarke2013functional}. The evolution of the state is continuous and bounded, the control is bounded, the agent's value is finite, the running cost is convex in the control, and the set of admissible effort paths is nonempty.} We verify the sufficiency of the necessary conditions in the proof of \cref{prop:characterisation} by showing that there is a unique solution to the necessary conditions. 

Using the notation from \Cref{sub:notation} the agent's objective is to maximize
\[\int_0^T e^{-\mu(T-\tau - A_{\tau})} (1-\bar{p} + \bar{p} e^{-\lambda A_{\tau}}) J(p_{\tau},\tau,a_{\tau}) d \tau\]

with $J(a_{\tau},A_\tau,\tau):= \mu (1-a_{\tau})  V(\tau) +\lambda  a_{\tau} p_{\tau} B - c$ the expected flow payoff at time $\tau:=T-t$.

Using $A_\tau$ as the state, the Hamiltonian at time $\tau$ is thus
  \begin{equation} \begin{split}
  H_\tau(a_\tau;A_\tau):=&e^{-\mu (T- \tau -A_\tau)} (1-\bar{p} + \bar{p}e^{-\lambda A_\tau}) J(p_\tau,\tau,a_\tau) + a_{\tau} \eta_\tau,\\
  = &e^{-\mu (T- \tau -A_\tau)} (1-\bar{p}) \left( (1-a_{\tau})\mu V(\tau) -c\right) \\ &+ e^{-\mu(T-\tau-\mu A_\tau)} \bar{p}e^{-\lambda A_\tau} \left((1-a_\tau)\mu V(\tau) + a_\tau \lambda B-c\right) + a_{\tau} \eta_\tau
  \end{split}\label{eq:derivativeeta}
\end{equation}
 where $\eta$ is the co-state. It has terminal condition $\eta_0=0$ and evolves according to\footnote{Note that because we take the derivative with respect to time remaining, the sign on the partial of the Hamiltonian is positive instead of negative.}
\begin{equation}
\begin{split}\frac{d \eta_\tau}{d \tau}:= 
&\frac{d H_\tau(a_\tau;A_\tau)}{d A_\tau}  \\
=& e^{-\mu(T- \tau-A_\tau)}\Bigg( \mu(1-\bar{p}) \Big((1-a_\tau) \mu V(\tau)-c\Big) \\
	& - (\lambda-\mu) e^{- \lambda A_\tau} \bar{p} \Big((1-a_\tau)\mu V(\tau) + a_\tau \lambda B-c\Big)\Bigg). 
\end{split}\label{eq:etaevol}
\end{equation}

Conditional on no arrival, the relative preference of thinking is determined by
\[ \begin{split}
\gamma_\tau &= -\frac{d H_\tau(a_\tau;A_\tau)}{d a_\tau}\\
&= e^{-\mu(T-\tau-A_\tau)} \left(\left(1-\bar{p}+\bar{p} e^{-\lambda A_\tau}\right)\mu V(\tau) - \bar{p} e^{-\lambda A_\tau}\lambda B \right)- \eta_\tau.
\end{split}\]

If $\gamma_\tau<0$, then the agent strictly prefers to do, $a_{\tau}=1$; if $\gamma_\tau>0$, then she strictly prefers to think, $a_\tau=0$.

\section{Key Lemmata} \label{sub:evolution_of_preference}

We state and prove four key lemmata. Combining these delivers most of our results up to and including parts of the characterization in \Cref{prop:characterisation}. The first, \cref{lem:dattheend}, states that the agent prefers to pull the doing arm close to the deadline independent of her belief about its quality. The second, \cref{lem:evolutiongamma}, states the evolution of the agent's relative preference between the arms over time. Because the Hamiltonian is linear in the agent's action, the evolution of the relative preference is independent of the agent's action. The third, \cref{lem:minimum}, states that the agent's relative preference has no interior minimum. This implies that the agent never \emph{returns} to the thinking arm if she stopped thinking without progress. The fourth, \cref{lem:thinkingbetterwithlonghorizons}, determines a condition such that the strategy ``doing throughout'' is dominated by the strategy ``think for a measure of time $dt$ before doing for the remainder of the time.''

\phantomsection
\addcontentsline{toc}{subsection}{Lemma \ref{lem:dattheend}}
\begin{lemma}\label{lem:dattheend}
	Suppose that the agent has not observed an arrival on either arm and holds belief $p\in(0,1)$ on the doing arm at the deadline. There is a remaining time $\hat{\tau}$ such that for the entire interval $\tau \in [0,\hat{\tau})$, the agent strictly prefers to pull the doing arm over pulling the thinking arm.
\end{lemma}
\begin{proof}
 $\gamma_\tau$ is continuous in $\tau$, and $\eta_0=0$. The value of a success on the thinking arm is continuous, and when $\tau=0$ it is $V(0)=0$. Thus, the terminal value of $\gamma$ is  
	\begin{equation*}
		\gamma_0 = - e^{-\mu(T-A_\tau)} p e^{-\lambda A_\tau} \lambda B <0.
	\end{equation*}
By the continuity of $\gamma_\tau$, there exists for any strategy $(a_\tau)_{\tau=0}^T$ a remaining time $\hat{\tau}>0$  such that $\gamma_{\tau}<0$ for $\tau<\hat{\tau}$, which proves the claim.\footnote{For any initial belief $\bar{p}\in(0,1)$ and any strategy $(a_\tau)_{\tau=0}^T$, the agent's terminal belief is in $(0,1)$.}
\end{proof}

\phantomsection
\addcontentsline{toc}{subsection}{Lemma \ref{lem:evolutiongamma}}
\begin{lemma}\label{lem:evolutiongamma}
A marginal increase in the time to the deadline $\tau$ changes the agent's policy function $\gamma_\tau$ by
\begin{equation}
\begin{split}
	\frac{d \gamma_\tau}{d \tau}&=e^{-\mu(T- \tau - A_\tau)} (1-\bar{p}+\bar{p} e^{-\lambda A_\tau})\left(\mu V'(\tau)+ p_\tau \mu \lambda (V(\tau)- B)+(\mu- \lambda p_\tau)c \right). 
\end{split}
\end{equation}
The change is independent of the agent's decision $a_\tau$.
\end{lemma}
\begin{proof}
First, recall that $H_{\tau}(a_{\tau};A_\tau)$ is affine in $a_\tau$, and note that for a function affine in $x$, $f(x;\theta)=t(\theta)+m(\theta)x$, it holds that 
\[\frac{d^2 f(x;\theta)}{d x d \theta} x=m'(\theta) x= \frac{d f}{d \theta} -\frac{d f}{d \theta}|_{x=0}.\]
Second, recall that $\frac{d \eta_\tau}{d \tau} = \frac{d H_\tau(a_\tau;A_\tau)}{d A_\tau}$ and that $-\frac{d A(\tau)}{d \tau}=a_\tau$, which yields (using $\gamma_\tau=-\frac{d H}{d a_\tau}$)

\begin{equation}
\begin{split}
\frac{d \gamma}{d \tau} &= -\frac{d^2 H_\tau(a_\tau;A_\tau)}{da_\tau d\tau} = -\frac{d^2 J(a_\tau;A_\tau,\tau)}{da_\tau d\tau}-\frac{d \eta}{d\tau}= -\frac{\partial}{\partial \tau}\frac{d J(a_\tau;A_\tau,\tau)}{d a_\tau} - \left.\frac{d \eta}{d \tau}\right|_{a_\tau=0}
\end{split}
\end{equation}
where we used that the law of motion of the state is independent of the state itself, $\frac{d^2J}{d a_\tau d \tau}= \frac{\partial}{\partial \tau}\frac{d J}{d a_\tau}-\frac{ d^2 H}{dA_\tau d a_\tau}a_\tau$ and $\frac{d \eta}{d\tau}-\frac{ d^2 H}{d A_\tau d a_\tau}a_\tau= \frac{d\eta}{d\tau}|_{a_\tau=0}$ based on the Hamiltonian being affine in $a_\tau$.

Third, from \eqref{eq:etaevol} we obtain
\begin{equation}
\begin{split}
\left. \frac{d \eta}{d \tau}\right|_{a_\tau=0} = e^{-\mu(T-\tau-A_\tau)}(1-\bar{p} +\bar{p}e^{-\lambda A_\tau}) \left(\mu V(\tau) -c\right)\left(\mu - \lambda p_\tau \right).
\end{split}
\end{equation}
Moreover, 
\[-\frac{\partial}{\partial \tau} \frac{ d J (a_\tau;A_\tau,\tau) }{d a_\tau} =  e^{-\mu(T-\tau-A_{\tau})}(1-\bar{p}+\bar{p} e^{-\lambda A_\tau})\left(\mu V'(\tau) +\mu\left(\mu V(\tau) - p_{\tau} \lambda B\right)\right)  \]
which implies 
\begin{equation}
\begin{split}
\frac{d \gamma}{d \tau} =&\phantom{=i}e^{-\mu(T-\tau-A_{\tau})}(1-\bar{p}+\bar{p} e^{-\lambda A_\tau})\left(\mu V'(\tau) +\mu\left(\mu V(\tau) - p_{\tau} \lambda B\right)\right) \\
&- e^{-\mu(T-\tau-A_\tau)}(1-\bar{p} +\bar{p}e^{-\lambda A_\tau}) \left(\mu V(\tau) -c\right)\left(\mu - \lambda p_\tau \right)\\
= &\phantom{=i}e^{-\mu(T-\tau-A_{\tau})}(1-\bar{p}+\bar{p} e^{-\lambda A_\tau})  \Bigg(\underbrace{\mu V'(\tau) + \mu p_\tau \lambda \left(V(\tau) - B\right) + \left(\mu - \lambda p_\tau\right)c}_{=:dy_\tau/d \tau}\Bigg).
\end{split}	\label{eq:delgammadeltau}
\end{equation}
 
Note that the sign of $d \gamma_\tau/d \tau$ is determined by the sign of $d {y}_\tau/d \tau$ only.
\end{proof}
The last line of \eqref{eq:delgammadeltau} implies that the agent's action has no first-order effect on the evolution of the switching function $y_\tau$. Any direct effect of the action on the instantaneous payoffs is counteracted by an effect on the continuation value. It is important, however, to keep in mind that the action has a second-order effect on $y_\tau$ through its effect on the evolution of the belief.

\phantomsection
\addcontentsline{toc}{subsection}{Lemma \ref{lem:minimum}}

\begin{lemma}\label{lem:minimum}
	Any minimum of $\gamma_\tau$ is either at $\tau=T$ or at $\tau=0$. Moreover, $y_{\tau}$ is strictly concave whenever $a_\tau=0$.
\end{lemma}

\begin{proof}
Since $\eta_\tau$ is continuously differentiable in $\tau$ and defined for any $\tau \in \mathbb{R}_+$, so is $y_\tau$. To prove \cref{lem:minimum}, we use that any interior (local) minimum has to be a critical point. From \cref{eq:delgammadeltau} we know that

\[\frac{d \gamma}{d \tau}=\underbrace{e^{-\mu(T- \tau - A_\tau)} (1-\bar{p}+\bar{p} e^{-\lambda A_\tau})}_{f(\tau)}\underbrace{\left(\mu V'(\tau)+ p_\tau \mu \lambda (V(\tau)- B)+(\mu- \lambda p_\tau)c \right)}_{\frac{d y_\tau}{d\tau}=:g(\tau)}.\]

If a critical point constitutes a local minimum, then it satisfies
\[ f(\tau) g(\tau)=0\qquad \text{ and }\qquad f'(\tau) g(\tau) + f(\tau) g'(\tau)>0\]
because $f(\tau)>0$ for all $\tau<\infty$, which implies $g(\tau)=0$, and any local minimum also requires $g'(\tau)=\frac{d^2 y_\tau}{d\tau d\tau}>0$. 
To show that such a local minimum cannot exist, we show that $g(\tau)=0$ and $g'(\tau)>0$ cannot be satisfied at the same time. It follows from $g(\tau)=0$ that
\[\mu \lambda (V(\tau)-B) = - \frac{\mu V'(\tau)+(\mu-\lambda p_\tau)c}{p_\tau}. \]
Differentiating $g(\tau)$ yields

\[\frac{d^2 y_\tau}{d \tau d \tau}= \frac{d p_\tau}{d \tau} \left( \mu \lambda (V(\tau) - B) - \lambda c\right) + \mu \left(V''(\tau) + \lambda p_\tau V'(\tau)\right) \]
and plugging in for $g(\tau)=0$ yields
\[\frac{d^2 y_\tau}{d \tau d \tau}|_{\frac{d y_\tau}{d\tau}=0}= - \frac{d p_\tau}{d \tau} \frac{\mu}{p_\tau}\left( V'(\tau)+c\right) + \mu \left(V''(\tau) + \lambda p_\tau V'(\tau)\right)\leq 0 \]
where the inequality follows as $\frac{d p_\tau}{d \tau}\geq 0, V'(\tau)>0$ and $-\frac{V''(\tau)}{V'(\tau)}\geq p_\tau \lambda$ as $p_\tau\leq p^0$.

The concavity of $y_\tau$ while the agent thinks follows straightforwardly by observing that $\frac{d p_\tau}{d \tau}=0$ in this case.
 
\end{proof}

\phantomsection
\addcontentsline{toc}{subsection}{Lemma \ref{lem:thinkingbetterwithlonghorizons}}

To state \cref{lem:thinkingbetterwithlonghorizons}, we define:
\[Z^d(p,\tau):= p\int_0^\tau e^{-\lambda t} (\lambda B-c) dt - (1-p) \int_0^\tau c dt=p U(\tau)-(1-p)c \tau,\]
and
\[ \begin{split}
Z^t(\varepsilon;p,\tau)&:= \int_0^\varepsilon e^{-\mu t}  (\mu V(\tau-t) - c) dt + e^{-\mu\varepsilon}Z^d(p,\tau- \varepsilon) \\&= \left(\mu V(\tau) - c\right) \varepsilon + (1{-} \mu \varepsilon) Z^d(p,\tau{-}\varepsilon){+}o(\varepsilon)
\end{split}\]
where the expression follows from a Taylor expansion around $0$.

The first, $Z^d$, describes the value absent a success of the strategy ``pull the doing arm from now until $\tau=0$'', given the belief $p$ and the time remaining $\tau$. The second, $Z^t$, describes the value absent a success of the strategy ``pull the thinking arm for a small measure of time $\varepsilon>0$, then pull the doing arm until $\tau=0$'', given the belief $p$ and the time remaining $\tau$.

\begin{lemma}\label{lem:thinkingbetterwithlonghorizons}
\[\lim_{\varepsilon \rightarrow 0} Z^t(\varepsilon;p,\tau)-Z^d(p,\tau)>0\] if and only if
\[\hat{q}(\tau):= \frac{\mu \left(V(\tau)+ c \tau\right)}{\mu\left(B + c \tau \right) + (\lambda-\mu)\left(B- U(\tau)\right)}>p.\]
Moreover for any $q \in (0,1)$ there is a $\tau$ such that $\hat{q}(\tau)= q$.
\end{lemma}

\begin{proof}

\begin{align*}
&&\lim_{\varepsilon \rightarrow 0} Z^t(\varepsilon;p,\tau) - Z^d(p,\tau)&>0\\
\Leftrightarrow&&\lim_{\varepsilon \rightarrow 0} \left(\mu V(\tau) - c\right) \varepsilon + (1{-} \mu \varepsilon) Z^d\left((p,\tau{-}\varepsilon) - Z^d(p,\tau)\right)+ \mu \varepsilon Z^d(p,\tau) &>0 \\
\Leftrightarrow&& \mu V(\tau) - c  -\mu Z^d(p,\tau) -  \lim_{\varepsilon \rightarrow 0} \frac{(Z^d(p,\tau) - Z^d(p,\tau-\varepsilon))}{\varepsilon} -   &>0 \\
\Leftrightarrow&& \mu V(\tau) - c  -\mu Z^d(p,\tau) -  \frac{ \partial Z^d (p,\tau)}{\partial \tau}   &>0 \\
\Leftrightarrow&& \mu V(\tau) - pU'(\tau) -pc  - \mu pU(\tau) + (1-p)\mu c \tau  &>0
\end{align*}
which is equivalent to
\begin{equation}\label{eq:q}
\hat{q}(\tau):=\frac{\mu \left(V(\tau)+ c \tau\right)}{\mu (U(\tau) + c \tau) + U'(\tau) +c}=\frac{\mu \left(V(\tau)+ c \tau\right)}{\mu(B + c \tau) + (\lambda-\mu)(B-U(\tau))}> p.
\end{equation}

The last claim follows because $\hat{q}(0)=0$, $\lim_{\tau \rightarrow \infty} \hat{q}(\tau)=1$ and $\hat{q}$ is continuous.

The limit $\tau \rightarrow \infty$ follows using L'H\^{o}pital's rule,

\[\lim_{\tau \rightarrow \infty} \hat{q}(\tau)= \lim_{\tau \rightarrow \infty} \frac{\mu(V'(\tau) +c)}{\mu(U'(\tau) +c)+U''(\tau)}=1\]
where the last equality follows from $\lim_{\tau \rightarrow \infty}U''(\tau)=\lim_{\tau \rightarrow \infty} U'(\tau)=0$ and $\lim_{\tau \rightarrow \infty} V'(\tau)=0$ because on an unbounded support any strictly concave, increasing, yet bounded function has to have a slope converging to zero if the limit of its derivative exists, which holds by assumption.
\end{proof}
\section{Proofs of Statements in the Main Text} \label{sec:proofs}

	\subsection{Proof of Proposition \ref{prop:Tinf}}
	\label{sub:proof_of_proposition_prop:Tinf}
	\begin{proof}
	When $T=\infty$, the value of an arrival on the thinking arm, $V(\infty)$, is constant over time. We can apply the standard dynamic programming approach for exponential bandits. The value function $u(p)$ given belief $p$ prior to termination satisfies\footnote{We obtain the expression using Taylor approximations for the success probabilities when $dt$ is small and using the ODE $dp/dt= p(1-p)a_\tau \lambda$. We do not require discounting as a single breakthrough on any arm ends the problem generating a finite payoff. Hence, the value is bounded even without discounting.}
	\begin{align*}
		u(p)=\max_{a\in [0,1]} (a p \lambda B + (1{-}a)\mu V(\infty) {-} c) dt + 1{-} (a p \lambda + (1{-}a)\mu)dt u(p){-}p(1{-}p)\lambda a dt u'(p).
	\end{align*}

	Letting $dt$ go to zero, dropping second-order terms and rearranging, we obtain the Bellman equation
	\begin{align*}
		0=\max_{a\in[0,1]} (a p \lambda B + (1-a)\mu V(\infty)) - c) - (a p \lambda + (1-a)\mu)u(p) - p(1-p)\lambda a  u'(p)
	\end{align*}
	where the maximand is linear in $a$. Whenever pulling the thinking arm is optimal at some time $t$, the agent chooses $a_t=0$, and it will remain optimal to pull $a_{t'}=0$ for all $t'>t$ as the belief remains constant. Thus, $u(p)=\int_0^\infty e^{-\mu t}(\mu V(\infty)-c) dt=V(\infty)-\frac{c}{\mu}$ whenever $a=0$. Whenever, pulling the doing arm is optimal, we can rewrite the Bellman equation as
	\begin{align*}
		0= p \lambda B - c - p \lambda u(p) - p (1-p)\lambda u'(p)
	\end{align*}
	and solving this differential equation yields 
	\begin{align*}
		u(p)= B - \frac{c}{\lambda} - (1-p)\left(\mathbb{C} + c \ln\left(\frac{p}{1-p}\right)\right)
	\end{align*}
	where $\mathbb{C}$ is a constant of integration. Using the value matching condition that the agent is indifferent between thinking and doing at $p=\hat{p}$, $u(\hat{p})=V(\infty)-\frac{c}{\mu}$, and the smooth pasting condition, $u'(\hat{p})=0$, we can obtain the constant of integration as well as $\hat{p}$, which are
	\begin{align*}
		\hat{p}&=\frac{\mu}{\lambda}\frac{c}{c+\mu(B-V(\infty))}\\
		\mathbb{C}&=B-\frac{c}{\lambda}\ln\left(\frac{\hat{p}}{1-\hat{p}}\right) - \left(V(\infty)-\frac{c}{\mu}\right).
	\end{align*}
	\end{proof}

	\subsection{Proof of Proposition \ref{prop:czero}}
	\label{sub:proof_of_proposition_prop:czero}

	\begin{proof}
		We make use of $d y_\tau/d \tau$ as defined in the proof of \cref{lem:evolutiongamma}. By \cref{lem:minimum}, the policy function $\gamma_\tau$ is twice continuously differentiable and has no interior minimum in $\tau$. This implies that if $d y_{\tau}/d \tau\geq 0$ for some $\tau$, then $d y_{\tau}/d \tau\geq 0$ for all $\tau' \in [0,\tau]$.
		Because condition \eqref{eq:benchmark_condition} holds, we obtain that for all $\tau \geq 0$
			\[\frac{d y_{\tau}}{d \tau}=\mu \left(V'(\tau) - \bar{p}\lambda (B-V(\tau))\right)\geq 0\]
		as $V(\tau)\leq B$ and $V'(\tau)\geq 0$. Thus, $y_{\tau}$ is increasing in the time remaining $\tau$ throughout. The agent pulls the doing arm close to the deadline by \cref{lem:dattheend}. What remains to be shown is if and when the agent switches from thinking to doing. Invoking \cref{lem:thinkingbetterwithlonghorizons} assuming $c=0$ yields that the agent never thinks (i.e., $\gamma_\tau <0$) if an only if $\forall \tau\leq T$

		\[\hat{q}(\tau|c=0):=\frac{\mu V(\tau)}{U'(\tau)+ \mu U(\tau)}<\bar{p}.\]
		Otherwise, she starts thinking and switches to doing with time $\tau_3$ remaining, where $\tau_3$ is the smallest solution to $\hat{q}(\tau_3|c=0)=\bar{p}$.
\end{proof}

\subsection{Proof of Proposition \ref{lem:ded}} \label{sub:proof_of_proposition_lem:ded}
\begin{proof}
The result follows from \cref{lem:dattheend,lem:minimum,lem:evolutiongamma}.

By \cref{lem:dattheend}, the agent pulls the doing arm shortly before the deadline whenever the game has not yet terminated. \Cref{lem:evolutiongamma} shows that $y_\tau$ is strictly concave at any critical point. Thus, $y_\tau\neq 0$ almost everywhere; i.e., the agent is generically not indifferent. Finally, by \cref{lem:minimum}, $y_\tau$ has no interior minimum: once the agent abandons the thinking arm, she does not return to it. Only the three strategies in \cref{lem:ded} remain possible.
\end{proof}

\subsection{Proof of Proposition \ref{prop:characterisation}} \label{sub:proof_of_proposition_prop:characterisation}
\begin{proof}
First, we show that any solution provided by the algorithm satisfies the necessary condition of the agent's optimal control problem. Second, we show that the algorithm provides a solution. Third, we show that there is a unique solution to the necessary conditions of the optimal control problem under \cref{ass:qincre}.

\paragraph{Step 1. The algorithm's solution is a candidate.} Here, we show that any solution to the algorithm satisfies the necessary conditions of the optimal control problem. We consider the different termination cases of the algorithm.

	\subparagraph{1a. The algorithm stops in \cref{step:onlyd}.} In this case, the algorithm's solution implies that the agent pulls the doing arm throughout. By \cref{lem:thinkingbetterwithlonghorizons}, it is optimal for the agent to follow this strategy as the agent will never be indifferent between thinking and doing.

	\subparagraph{1b. The algorithm stops in \cref{step:ed}.} \cref{lem:thinkingbetterwithlonghorizons} implies that it is optimal for the agent to pull the doing arm for the final remaining time $\tau<\overline{\tau}_3$ when she holds belief $\bar{p}$ at $\tau_3$. The function $\dot{y}(\tau;q(\overline{\tau}_3),\overline{\tau}_3)$ has the same sign as the slope of the agent's policy function from \cref{lem:evolutiongamma} when pulling the thinking arm conditional on the agent pulling the doing arm for any remaining time $\tau<\overline{\tau}_3$. The next Lemma shows that $\dot{y}(\tau;q(\overline{\tau}_3),\overline{\tau}_3)\geq 0$.

	\phantomsection

\addcontentsline{toc}{subsubsection}{Lemma \ref{lem:dotypositive}}
		\begin{lemma}\label{lem:dotypositive}
			Suppose that it is optimal for the agent to switch from pulling the thinking arm for a positive measure of time to pulling the doing arm with time $\tau_3$ remaining. Then 
			\[\dot{y}(0;q({\tau}_3),{\tau}_3)> 0.\]
		\end{lemma}
		\begin{proof}
			To the contrary, assume that $\dot{y}(0;q({\tau}_3),{\tau}_3)\leq 0$. Because switching to the doing arm is optimal with time $\tau_3$ remaining, by \cref{lem:ded} and $\gamma_\tau$ being a continuously differentiable function, we must have that $\gamma_{\tau_3}=0$. With $\dot{y}(0;q({\tau}_3),{\tau}_3) < 0$, this implies that there is an $\overline{\varepsilon}$ such that $\gamma_{\tau_3-\varepsilon}>0$ for all $\varepsilon\in (0,\overline{\varepsilon})$, which implies (strict) optimality of thinking with time $\tau_3-\varepsilon$ remaining, a contradiction to switching to the doing arm with time $\tau_3$ remaining. With $\dot{y}(0;q({\tau}_3),{\tau}_3) = 0$, the agent would pull the doing arm immediately again for $\tau>\tau_3$, as any critical point corresponds to a strict local maximum, which is a contradiction to thinking for a positive measure of time before $\tau_3$.
		\end{proof}

		By \cref{lem:dotypositive}, there is a $\Delta>0$ such that for all $\tau \in (\overline{\tau}_3,\overline{\tau}_3+\Delta]$, $\hat{y}(\tau;\bar{p},\overline{\tau_3})>0$, as $\hat{y}$ is continuous in $\tau$. Thus, $\tau_2(\overline{\tau_3})$ is defined, and for any time horizon $T<\tau_2(\overline{\tau_3})+\overline{\tau_3}$, it is optimal for the agent to start thinking before switching to the doing arm for the remaining time.

	\subparagraph{1c. The algorithm stops in \cref{step:convergence}.} The continuation game with time $\tau_2(\overline{\tau_3})+\overline{\tau_3}$ is identical to one in which the agent starts out with belief $\bar{p}=q(\overline{\tau_3})$ and deadline $T=\tau_2(\overline{\tau_3})+\overline{\tau_3}$. Moreover, $\tau_1(\tau_3)$ describes the length of the initial doing period to arrive at this continuation game. By construction, at $\tau_2(\tau_3)+\tau_3$, the agent's policy function $\gamma$ must be increasing, as it coincides with $\hat{y}$ multiplied by a positive constant. By \Cref{lem:minimum}, $\gamma_\tau$ cannot have an interior minimum and, therefore, must be negative for all $\tau\in(T-\tau_2(\tau_3)-\tau_3,T]$.

\paragraph{Step 2. The algorithm finds a solution.} Here, we show that the algorithm always provides a solution. First, we state a lemma that will be useful for the remainder.

\phantomsection
\addcontentsline{toc}{subsubsection}{Lemma \ref{lem:monotonicityqtau12}}
	\begin{lemma}\label{lem:monotonicityqtau12} Under \Cref{ass:qincre}, the following monotonicity statements hold.
		\begin{itemize}
			\item[(i)] $q(\tau_3)$ is monotonically increasing in $\tau_3$.
			\item[(ii)] $\tau_1(\tau_3)$ is monotonically decreasing in $\tau_3$.
			\item[(iii)] $\tau_2(\tau_3)$ is monotonically decreasing in $\tau_3$.
		\end{itemize}
	\end{lemma}
	\begin{proof}
		We prove each statement separately. Recall that $q(\tau)=\min\{1,\hat{q}(\tau)\}$.
		\paragraph{Statement (i).} To simplify exposition, we use the notation $\hat{q}(\tau)=\frac{x}{z}$, with $x=\mu(V(\tau)+ c \tau)$ and $z=\mu(U(\tau)+c\tau)+U'(\tau)+c$. Note that $q(\tau)$ is continuous and that $\hat{q}'(\tau=0)>0$. Thus, if $\hat{q}(\tau)$ ever decreases a local maximum must exist. Moreover, $\lim_{\tau \rightarrow \infty} \hat{q}(\tau)=1$. Thus, if $q(\tau)$ ever decreases, $\hat{q}(\tau)$ must have a local minimum with $\hat{q}(\tau)<1$.  We will show that $q(\tau)$ cannot be decreasing by showing that $\hat{q}(\tau)$ has no local minimum with $\hat{q}(\tau)<1$.

			\subparagraph{Case 1: $\mu\leq \lambda$.} Consider $\mu\leq \lambda$. If $\hat{q}(\tau)$ is ever decreasing at least one local maximum exists such that $\hat{q}'(\tau)=0$ and $\hat{q}''(\tau)<0$. At this maximum we thus have, 
				\begin{equation}
					\frac{V'(\tau)+c}{V(\tau)+c\tau}=\frac{(\mu-\lambda)U'(\tau)+\mu c}{\mu(U(\tau)+c\tau)+U'(\tau)+c}
				\end{equation} and 
				\begin{equation}
					\frac{V''(\tau)}{V'(\tau)+c}<\frac{(\mu-\lambda)U''}{(\mu-\lambda)U'(\tau)+\mu c},\label{eq:socmulambda}
				\end{equation}
				where we have used $U''=-\lambda U'$ and that $\hat{q}'(\tau)=0$ implies $\frac{x'}{x}=\frac{z'}{z}$. A critical point requires $\frac{x'}{x}=\frac{z'}{z}$, and since $x',x$ and $z$ are trivially greater than zero, so is $z'$. It follows that whenever $\mu<\lambda$, the right-hand side of \eqref{eq:socmulambda} is positive, implying that any critical point must be a local maximum. Indeed, the right-hand side is positive: The numerator is positive, because $U''(\tau)<0$ and $\mu\leq 0$ and the denominator is positive as well by the above argument that $z'>0$. Because $\lim_{\tau \rightarrow \infty }\hat{q}(\tau)=1$, any local maximum at $\hat{\tau}$ with $\hat{q}(\hat{\tau})<1$ would imply the existence of a local minimum for some $\check{\tau}>\hat{\tau}$. Hence, if $\hat{q}(\tau)$ is decreasing for some $\tau'$ we must have that $\hat{q}(\tau)>1$ for all $\tau \geq \tau'$, implying that $q(\tau)=1$ for all $\tau \geq \tau'$ proving monotonicity of $q(\tau)$.

			\subparagraph{Case 2: $\mu> \lambda$.} Consider $\mu>\lambda$. Recall that whenever $\hat{q}'(\check{\tau})<0$ and $\hat{q}(\check{\tau})<1$, there must be some local minimum; denote the time remaining at the local minimum by $\check{\tau}_2$. As $\hat{q}'(0)>0$, there must be a local maximum of $\hat{q}(\tau)$ first; denote the time remaining at the local maximum by $\check{\tau}_1$, with $\check{\tau}_1<\check{\tau}_2$. Moreover, as $\lim_{\tau \rightarrow \infty} \hat{q}(\tau)=1$, there must be either another local maximum, the time remaining of which is denoted by $\check{\tau}_3$, or there is some $\tilde{\tau}$ such that for all $\tau>\tilde{\tau}$, $\hat{q}''(\tau)<0$. Define $\varphi(\tau):=x''z-z''x$.

			These three observations imply the following: $\varphi(\check{\tau}_1)<0$; $\varphi(\check{\tau}_2)>0$; if $\check{\tau}_3$ exists, then $\varphi(\check{\tau}_3)<0$; and if $\check{\tau}_3$ does not exist, then there must be some $\tilde{\tilde{\tau}}$ such that $\varphi(\tilde{\tilde{\tau}})<\varphi(\check{\tau}_2)$. The latter conclusion follows by the observation that $\hat{q}''(\tau)<0$ whenever $\hat{q}(\tau)$ converges from below to 1. $\hat{q}''(\tau)<0$ implies that $\varphi(\tau)<2\frac{z'}{z}(x'z-xz')$, where the right-hand side converges to 0 as $\tau \rightarrow 1$. 

			Thus, we know that as $\tau$ moves from $\tau\leq \check{\tau}_1$ to $\infty$, $\phi(\tau)$ is strictly negative (at $\check{\tau}_1$), strictly positive (at $\check{\tau}_2$) and arbitrarily small as we approach $\tau=\infty$. Thus, $\phi(\tau)$ has to be nonmonotonic. Part (ii) of \Cref{ass:qincre} rules nonmontonicity out. Thus, no local minimum of $\hat{q}(\tau)$ with $\hat{q}(\tau)<1$ exists, and $q(\tau)$ is monotonic.

		\paragraph{Statement (ii).} The monotonicity of $\tau_1$ follows by the monotonicity of $q(\tau_3)$ and the observation that $\tau_1(\tau_3)$ decreases in $q(\tau_3)$.

		\paragraph{Statement (iii).}	To see that $\tau_2(\tau_3)$ decreases, recall that $\tau_2$ is determined via the root of $\hat{y}(\tau;q(\tau_3),\tau_3)$ whenever this root exists for some $\tau>0$. In this case, we require by definition of $\tau_2(\tau_3)$\footnote{Recall that because $\bar{p}<1$, $q(\tau)$ is differentiable in the relevant part.} 

			\begin{equation}\label{eq:root}
				\begin{split}
					\frac{\mathrm d \hat{y}(\tau_2(\tau_3);q(\tau_3), \tau_3)}{\mathrm d \tau_3} = &\frac{\partial \hat{y}(\tau_2(\tau_3);q(\tau_3), \tau_3)}{\partial \tau_3} + \frac{\partial \hat{y}(\tau_2(\tau_3);q(\tau_3), \tau_3)}{\partial q(\tau_3)} \frac{\partial q(\tau_3)}{\partial \tau_3}\\  
										&+ \frac{\partial \hat{y}(\tau_2(\tau_3);q(\tau_3), \tau_3)}{\partial \tau_2} \frac{\partial \tau_2(\tau_3)}{\partial \tau_3} = 0.
				\end{split}	
			\end{equation}

			Note that under \cref{ass:qincre}, $\frac{d}{d\tau_3} \dot{y}(s;q(\tau_3),\tau_3)<0$ because
			\[\frac{d\dot{y}(s;q(\tau_3),\tau_3)}{d\tau_3} = \mu V''(s+\tau_3) +\mu \lambda \left( q(\tau_3) V'(s+\tau_3)  +\frac{d q(\tau_3)}{d\tau_3} \left(V(s+\tau_3) - B - \frac{c}{\mu}\right)\right)\]
			which is negative for all $s$, as $-V''(\tau_3+s)/V'(\tau_3+s)\geq \bar{p}\lambda \geq q(\tau_3)\lambda$ and $V(\tau_3+s)\leq B+\frac{c}{\mu}$ by assumption while $q(\tau_3)$ is increasing in $\tau_3$.

			Hence, we know that
			\[\frac{\partial \hat{y}(\tau_2(\tau_3);q(\tau_3), \tau_3)}{\partial \tau_3} + \frac{\partial \hat{y}(\tau_2(\tau_3);q(\tau_3), \tau_3)}{\partial q(\tau_3)} \frac{\partial q(\tau_3)}{\partial \tau_3} <0\] 

			and, moreover, that $\partial \hat{y}(\tau=\tau_2;p,\tau_3)/\partial \tau<0$ because $\tau_2$ is the root and $\hat{y}(\tau;p,\tau_3)>0$ if $\tau<\tau_2$ by construction. Thus, to satisfy \eqref{eq:root}, we need $\partial \tau_2(\tau_3)/\partial \tau_3<0$.
	\end{proof}

Second, we show that if there is a solution without an initial doing period, then the algorithm always returns such a solution. 

If such a solution exists, then there exists a $\overline{\tau}_3$ such that $\bar{p}=q(\overline{\tau}_3)$, as $\lim_{\tau \rightarrow \infty} q(\tau)=1>\overline{p}$, which \cref{step:onlyd} of the algorithm will detect because $q$ is monotonic by \cref{lem:monotonicityqtau12}. If the solution is such that only a Hail Mary period is possible, then \cref{step:onlyd} ensures that $\tau_3=T$, as the solution detected $\overline{\tau}_3>T$. If a solution in which the agent starts by thinking is possible, then \cref{step:ed} detects one such solution, i.e., if $\overline{\tau}_3<T$ from \cref{step:onlyd} and $q(\overline{\tau}_3)=\overline{p}$. Neither \cref{step:onlyd} nor \cref{step:ed} returns a solution only if any policy that involves only a single doing period does not satisfy the necessary conditions of the optimal control problem.

Third, we show that if all solutions involve an initial doing period, then \cref{step:convergence} of the algorithm finds such a solution.

If an initial doing period exists, then the belief held at the beginning of the Hail Mary period must satisfy $q(\tau_3)<\bar{p}$. At the same time, by \cref{lem:dattheend}, $q(\tau_3)>0$, and 
\[q(\tau_3-t) \leq \frac{ q(\tau_3) e^{-\lambda (\tau_3-t)}}{q(\tau_3) e^{-\lambda (\tau_3-t)} + 1-q(\tau_3)}\]
for all $t \in [0,\tau_3]$ such that \cref{lem:thinkingbetterwithlonghorizons} does not imply any additional switches for any time remaining $\tau<\tau_3$.

By \cref{lem:minimum,lem:evolutiongamma}, any solution with an initial doing period implies the existence of two roots of the policy function $\gamma_\tau$. Because $\hat{y}(\tau;q(\tau_3),\tau_3)=0$ determines the smallest root $\tau>0$ of $\gamma_\tau$ conditional on a Hail Mary period of length $\tau_3$, the algorithm detects that root if it exists. Finally, because $q(\tau_3)<\bar{p}$ and beliefs are constant while thinking, the length of the initial doing period is determined by Bayes' rule and the belief conditional on reaching the Hail Mary period $q(\tau_3)= \bar{p} e^{-\lambda t}/(\bar{p} e^{-\lambda t}+1-\bar{p})$, which results in $\tau_1(\tau_3)$.

\Cref{step:convergence} of the algorithm considers all possible combinations of $\tau_1(\tau_3),\tau_2(\tau_3)$ and $\tau_3$ until a solution is found that satisfies the necessary conditions. If a solution exists, the algorithm converges.

Fourth and finally, a solution to the optimal control problem exists because the evolution of the state is continuous and bounded, the control is bounded, the agent's value is finite, the running cost is convex in the control and the set of admissible effort paths is nonempty (see, e.g., \citet{clarke2013functional} for details). By \cref{lem:ded}, any solution is of one of the three types the algorithm considers. Thus, the algorithm determines a candidate solution.

\paragraph{Step 3. The algorithm's solution is the unique candidate.} Finally, we show that the algorithm identifies the uniquely optimal policy. To show uniqueness, we have to show that given a solution $\tau_3$, there is no other $\tau_3'\neq \tau_3$ that solves the fixed point problem.
	
	Any two solutions in which the agent pulls the doing arm for only one time interval are identical on any positive measure of time. This is immediate because if the agent pulls the doing arm on only one interval, then she has to pull it in the end. Either the agent pulls only the doing arm in which the strategy is trivially unique and $q(\tau)<p_\tau$ for any $\tau \in[0,T]$ or she begins by pulling the thinking arm. In the latter case, she switches when the time remaining is $q^{-1}(\bar{p})$, which has a unique solution by \Cref{lem:monotonicityqtau12}, as $\bar{p}<1$.

	Thus, if there are two candidate strategies satisfying the necessary conditions, then the agent needs to split the time spent on the doing arm between two disjoint intervals in at least one of those strategies. If the agent splits her time doing in at least one solution and the two solutions differ on a positive measure of time, then there must be two different lengths of the Hail Mary period, $\tau_3'$ and $\tau_3$, both of which satisfy the necessary conditions. Assume without loss of generality that $\tau_3'>\tau_3$. Both $\tau_3'$ and $\tau_3$ have associated terminal beliefs, $\underline{p}'$ and $\underline{p}$. The terminal belief is the agent's belief at the deadline conditional on failing to find any solution.  Note that for the case of $\tau_3$, the agent's strategy must involve two distinct doing periods. We proceed by cases and derive a contradiction for each of them.

	\subparagraph{Assume $\underline{p}>\underline{p}'$.} Consider the agent's belief with $\tau_3$ periods remaining, and assume that she pulled the doing arm in the interval $[\tau_3',\tau_3)$ with initial belief $q(\tau_3')$. Since $\underline{p}>\underline{p}'$, the agent has to hold a belief $\tilde{p}(\tau_3)<q(\tau_3)$ with $\tau_3$ periods remaining. However, then the agent prefers to pull the thinking arm with $\tau_3$ periods remaining by \cref{lem:thinkingbetterwithlonghorizons}, which is a contradiction.

	\subparagraph{Assume $\underline{p}=\underline{p}'$}. Consider the agent's policy function under the strategy that implies the last switch to occur at $\tau_3'$: $\gamma'_\tau$. The necessary conditions imply that $\gamma'_{\tau_3'}=0$. Because the terminal beliefs coincide, the policy function and hence the strategy in the continuation game for $\tau<\tau_3$ coincide with the policy function and the strategy corresponding to a Hail Mary period of length $\tau_3$ only---as the terminal condition $\gamma_0$ depends only on the terminal belief. In turn, this observation implies that $\gamma'_{\tau_3}=\gamma_{\tau_3}=0$. However, by construction, the agent pulls the doing arm with time remaining $\tau=[\tau_3',\tau_3]$, implying that $\gamma'_\tau\leq 0$ on this interval. As a consequence, $\gamma'_{\tau}$ has to have a critical point at $\tau_3$. The arguments in the proof of \cref{lem:minimum} imply that $\gamma'_\tau$ is strictly concave at any critical point, and thus, $\gamma'_\tau$ attains a maximum at $\tau_3$. As the beliefs coincide at $\tau_3$, the policy functions under both strategies attain a maximum at $\tau_3$. By \cref{lem:minimum}, none of the policy functions will attain a maximum, and thus, the agent pulls the doing arm throughout under both policy functions, contradicting the assumption that the strategies differ and that there is a switching time $\tau_3'>\tau_3$.
	
	\subparagraph{Assume $\underline{p}<\underline{p}'$.} In this case, the agent's overall time spent on the doing arm must be smaller with switching time $\tau_3'$ than with $\tau_3$. This implies that both strategies involve two distinct doing periods. Moreover, $\tau_2(\tau_3')>\tau_2(\tau_3)$ for both $\tau_3$ and $\tau_3'$ to be a solution to the fixed point problem. By \cref{lem:monotonicityqtau12}, $\tau_2$ decreases in $\tau_3$, which is a contradiction.
\end{proof}

\subsection{Proof of Proposition \ref{prop:doing_hurts_success}}
\begin{proof}
The probability that the agent obtains a success before the deadline for any $\tau_1,\tau_2$, and $\tau_3$ is 
\begin{align}
	P(\tau_1,\tau_2,\tau_3;T)&=\bar{p} (1-e^{-\lambda \tau_1}) + (\bar{p} e^{-\lambda \tau_1}+1-\bar{p})\\ 
	&\cdot \left(1-e^{-\mu \tau_2}-\frac{e^{-\nu T}-e^{-\mu \tau_2-\nu(T-\tau_2)}}{\mu-\nu}+e^{-\mu \tau_2} \frac{\bar{p}e^{-\lambda \tau_1}}{\bar{p} e^{-\lambda \tau_1}+1-\bar{p}} (1-e^{-\lambda \tau_3})\right)
\end{align}
where $\tau_3=T-\tau_1-\tau_2$.
Consider the derivative of $P(\tau_1,\tau_2,\tau_3;T)$ with respect to $\tau_1$, which is
\begin{align}
	e^{-\nu(T-\tau_1)-\mu \tau_2-\lambda \tau_1}\frac{e^{\mu \tau-2}-e^{\nu \tau_2}}{\mu-\nu}\mu(\bar{p}(\lambda - \nu)-(1-\bar{p})e^{\lambda \tau_1}\nu ).
\end{align}
Its sign is determined by the sign of the last term, which is negative whenever $\nu\geq p_{\tau_1}\lambda$. This condition is satisfied by our assumptions; in particular, it is a consequence of the relative concavity assumption.
\end{proof}

\subsection{Proof of Proposition \ref{lem:tau_cs}} \label{sub:proof_of_proposition_lem:belief_declines}

\begin{proof}
	Suppose towards a contradiction that $\tau_1$ decreases in $T$. In particular, consider two scenarios: (i) deadline $T$ and (ii) deadline $T'>T$. Moreover, suppose that the associated initial doing periods are such that $\tau'_1<\tau_1$.

	The belief $q(\tau_3')$ that the agent holds during the thinking period in scenario (ii) is larger than the belief $q(\tau_3)$ the agent holds in scenario (i). By \cref{lem:thinkingbetterwithlonghorizons,lem:monotonicityqtau12}, $q(\cdot)$ is monotonic and increasing in $\tau_3$, which implies $\tau'_3>\tau_3$.

	Consider both scenarios with time $\tau_3$ remaining. In scenario (i), the agent is indifferent between thinking and doing by construction and prefers doing for the remainder of the time. In scenario (ii), she prefers doing at $\tau_3$ and for the remainder of time because $\tau_3'>\tau_3$. This implies for the corresponding beliefs at remaining time $\tau_3$ that $p'_{\tau_3}\geq p_{\tau_3}=q(\tau_3)$. For any subsequent period, the agent pulls the doing arm in both scenarios. It follows that the terminal beliefs are $\underline{p}'\geq \underline{p}$.

	Because in both scenarios the agent starts with a belief $\bar{p}$, a larger terminal belief in scenario (ii) implies that the maximum time the agent pulls the doing arm decreases in this scenario compared to scenario (i); i.e., $\tau_1+\tau_3 \geq \tau'_1+\tau'_3$. Because $T'>T$, it follows that $\tau'_2>\tau_2$. By \cref{lem:monotonicityqtau12}, $\tau'_3\geq \tau_3$ implies $\tau'_2 \leq \tau_2$. In addition, $\frac{d \dot{y}}{d p}<0$, implying that $\tau_2$ decreases in the belief as well. Thus, $\tau'_3\geq \tau_3$ and $q(\tau'_3)\geq q(\tau_3)$ imply $\tau'_2\leq \tau_2$, which is a contradiction.

	The length of the Hail Mary period increases in $T$ if the agent immediately enters this period and is constant whenever the agent starts by thinking. Finally, because $\tau_1$ is nondecreasing, it follows that $q(\tau_3)$ is nonincreasing, and thus, by \cref{lem:monotonicityqtau12}, $\tau_3$ nonincreasing.

	Whenever the agent has no initial doing period, $\tau_2$ is trivially nondecreasing in $T$. Because $\tau_3$ is nonincreasing when $\tau_1>0$, it follows from \cref{lem:monotonicityqtau12} that $\tau_2$ weakly increases. Finally, $\tau_1+\tau_3$ is bounded because $q^{-1}(\tau_3)$ is bounded by $\hat{p}>0$ defined in \cref{prop:Tinf}, which in turn implies that both $\tau_1$ and $\tau_3$ are bounded. However, then, because a solution exists for every $T$, we must have that $\tau_2 \rightarrow \infty$ as $T \rightarrow \infty$.
\end{proof} 
\subsection{Proof of Proposition \ref{prop:beliefs}} \label{sec:proof_of_proposition_prop:beliefs}

\begin{proof}
	We prove each item separately. At several points in the proof, we invoke \cref{lem:monotonicityqtau12,lem:dotypositive}, which can be found in \cref{sub:proof_of_proposition_prop:characterisation} \nameref{sub:proof_of_proposition_prop:characterisation}.

	\paragraph{Proof of \cref{it:overlinep}.}

	$\tilde{p}$ is constructed such that $\dot{y}(0,\tilde{p},q(\tilde{p}))=0$. Thus, if $\bar{p}>\tilde{p}$, then \[\underbrace{\mu V'(q^{-1}(\bar{p})) - \bar{p}\lambda \mu (B-V(q^{-1}(\bar{p}))) + (\mu-\lambda \bar{p})c}_{=\dot{y}(0,\bar{p},q(\bar{p}))}<0,\]
	 because $q^{-1}(p)$ is monotonic for $p \in (0,1)$ and the LHS is decreasing since $-V''(\tau)/V'(\tau)>\lambda \bar{p}$ by the relative concavity assumption. Using \cref{lem:dotypositive}, this implies that it cannot be optimal to switch from thinking to doing with $q^{-1}(\bar{p})$ remaining or, equivalently, with a belief $\bar{p}>\tilde{p}$. 

	\paragraph{Proof of \cref{it:hatp}.} Observe that the sign of function $\dot{y}(s;p,\xi)$ is the same as the sign of function $\frac{d\gamma_\tau}{d\tau}$ conditional on $\gamma_\xi=0$ and an agent that thinks with time remaining $\tau \in [\xi,s+\xi]$. A necessary condition for an initial doing period is that there exists a $\tau_2>0$ such that $y(\tau_2;q(\tau_3),\tau_3)=0$ with $\tau_3=q^{-1}(p)$ and $p<\bar{p}$ the belief held during the thinking period. Because the agent (i) is indifferent with time $\tau_2+\tau_3$ remaining, (ii) pulls the doing arm for a positive measure of time before that and (iii) expects to pull the thinking arm for a positive measure of time thereafter, $\dot{y}(\tau_2;p,\tau_3)<0$ by the same arguments that proved \cref{lem:dotypositive}. Because $- V''(\tau)/V'(\tau)> p \lambda V'$ by assumption, $\dot{y}(\cdot)$  decreases in $\tau_2$. Thus, if $\lim_{\tau_2\rightarrow \infty} \dot{y}(\tau_2;p, \tau_3)\geq 0$, then there is no second root of $y(\tau_2;p,\tau_3)$ for any $\tau_2>0$. It follows that if $\bar{p}<\hat{p}$, then
	\[
	\lim_{\tau_2 \rightarrow \infty} \dot{y}(\tau_2; p, \tau_3)\geq 0 \Leftrightarrow p < \frac{V'(\tau_2+\tau_3)+c}{\lambda \left(B+\frac{c}{\mu}-V(\tau_2+\tau_3)\right)}= \frac{c}{\lambda \left(B+\frac{c}{\mu} - V(\infty)\right)}=\hat{p}.
	\]

	\paragraph{Proof of \cref{it:underlinep}.} First, we show that if $\bar{p} \geq \hat{p}$, then the agent enters the Hail Mary period with a belief $p_{\tau_3}\geq\hat{p}$. Suppose towards a contradiction that the agent switches to the Hail Mary period with time $\tau_3$ remaining and a belief $p_{\tau_3}<\hat{p}$. Because $\bar{p}\geq\hat{p}$, there has to be a continuation game with time $\tau$ remaining at which the agent is in the initial doing phase and holds a belief $p_\tau<\hat{p}$. By \cref{it:hatp} such a continuation game cannot exist.

	Second, observe that if the agent enters the Hail Mary period with a belief $\hat{p}$, then her terminal belief at the end of an unsuccessful Hail Mary period is $\check{p}$. 

	Third, suppose that the terminal belief was smaller than $\check{p}$. Then, it must be true that with time $q^{-1}(\hat{p})$ remaining, the agent is in the Hail Mary period because the agent enters the final doing period with a belief weakly greater than $\hat{p}$ and because $q^{-1}(\hat{p})$ is the time length of pulling the doing arm required to deteriorate a belief of $\hat{p}$ to $\check{p}$. However, that the terminal belief lies below $\check{p}$ implies that with time $q^{-1}(\hat{p})$ remaining, the agent's belief is less than $\hat{p}$. However, this contradicts the necessary conditions for an optimal strategy, as $q(q^{-1}(\hat{p}))>p_{q^{-1}(\hat{p})}$, which implies that the agent is not in the Hail Mary period.
\end{proof}

 \setstretch{1}
\part{References}
\begin{refcontext}[sorting=nyt]
\printbibliography[heading=none]
\end{refcontext}\part{Supplementary Material} \label{sub:no_shirking_condition}
In this part, we verify that the examples considered in the main text satisfy our assumptions, consider the extension to the value of progress being dependent on the state $A_\tau$, consider an example in which the relative concavity assumption does not hold, and provide a sufficient condition such that the agent never wants to shirk.

\section{Verification of Assumptions for Examples} \label{sub:ex_verify}

In  this part, we verify that the examples discussed in \Cref{ssub:examples} satisfy the assumptions of our model.

\subsection{Verification of Example \ref{example_implications} \& Example \ref{par:ex_delay}}
The value of such an arm given remaining time $\tau$ is 
\begin{align*}
	V(\tau)= (1-e^{-\nu \tau})\left(B_\nu-\frac{c_\nu}{\nu}\right),
\end{align*}
with
\begin{align*}
	V(0)=0,~V'(\tau)=e^{-\nu \tau}\left(\nu B_\nu-c_\nu\right),~V''(\tau)= - \nu e^{-\nu \tau}\left(\nu B_\nu-c_\nu\right),~	-\frac{V''(\tau)}{V'(\tau)}= \nu.
\end{align*}

\subsection{Verification of Example \ref{par:ex_risky}}
As with any risky bandit, the agent would have an incentive to eventually stop pulling the arm if her time were unlimited. Assume that this occurs after pulling the arm for $\hat{t}$ periods.
	\[V(\tau)= \begin{cases}
		\bar{p}^\nu(1-e^{-\nu \tau})\left(B_\nu-\frac{c_\nu}{\nu}\right)-(1-\bar{p}^\nu)c_\nu \tau &, \text{ if }\tau\leq \hat{t} \\
		\bar{p}^\nu B_\nu - \frac{c_\nu}{\nu} (1+(1-\bar{p}^\nu) \hat{t}) &, \text{ if } \tau \geq \hat{t},
		\end{cases}
		\]

where $\hat{t}$ is the time at which an agent with initial belief $\bar{p}^\nu$ would stop experimenting on the new arm. Such an arm satisfies our desired conditions when the initial belief is sufficiently high given the deadline, i.e., whenever $\tau\leq \hat{t}$:\footnote{Note that $\hat{t}$ increases in $\bar{p}^v$ and converges to $\infty$ as $\bar{p}^v \rightarrow 1$, the case in which \cref{par:ex_risky} converges to \cref{example_implications}.}
\begin{align*}
	V(0)&=0\\
	V'(\tau)&=
		\bar{p}^\nu e^{-\nu \tau}\left(\nu B_\nu-c_\nu\right)-(1-\bar{p}^\nu)c_\nu  \\
			V''(\tau)&= - \nu \bar{p}^\nu(e^{-\nu \tau})\left(\nu B_\nu-c_\nu\right) \\
		-\frac{V''(\tau)}{V'(\tau)}&=\nu \frac{\bar{p}^\nu(e^{-\nu \tau})\left(\nu B_\nu-c_\nu\right)}{\bar{p}^\nu e^{-\nu \tau}\left(\nu B_\nu-c_\nu\right)-(1-\bar{p}^\nu)c_\nu}>\nu.
\end{align*}	

\subsection{Verification of Example \ref{par:ex_timevarying}}
In this version of the model, the value does not have a closed form solution, but we can verify that our assumptions are satisfied whenever $\beta<\hat{\beta}:=e^\alpha \nu+\frac{\bar{p} \lambda}{\nu}\frac{ce^{-\alpha}-\nu B}{\nu B}-\frac{c}{B}$.\footnote{Depending on parameters, this bound can be either positive or negative. In particular, it is strictly positive whenever $\nu$ is sufficiently high.}
\begin{align*}
	V(\tau)&=\int_0^\tau e^{-\int_0^t \nu e^{\alpha+\beta s} ds}\left(\nu e^{\alpha+\beta t} B -c \right) dt<B\\
	V(0)&=0\\
	V'(\tau)&=e^{-\nu e^{\alpha+\beta \tau} \tau}\left(\nu e^{\alpha+\beta \tau }B - c\right)>0 \\
	V''(\tau)&=-\nu e^{\alpha+\beta \tau-e^{\alpha+\beta \tau}}\left((1+\beta \tau) (\nu e^{\alpha+\beta \tau}B-c-\beta B) \right)<0\\
	-\frac{V''(\tau)}{V'(\tau)} &=  e^{\alpha+\beta \tau} (1+\beta \tau)-\beta \frac{\nu B e^{\alpha+\beta \tau}}{\nu B e^{\alpha+\beta \tau}-c}>\bar{p} \lambda.
\end{align*}
To see the sign of $V''(\tau)$, note that the term in parentheses is increasing in $\tau$ and positive for $\tau=0$, which determines the sign of the second derivative of the value. To see the sign of $-\frac{V''(\tau)}{V'(\tau)}$, note that this expression is also increasing in $\tau$. Evaluating $-\frac{V''(0)}{V'(0)}$ delivers the desired expression for $\hat{\beta}$, which ensures that our conditions are satisfied.

\subsection{Verification of Example \ref{par:ex_stream}}
 The expected value at time $\tau$ of $b(t)$, where $b(t)$ follows an Ornstein-Uhlenbeck process, yields the benefit of triggering this payoff stream
\begin{align*}
	V(\tau)&=b(0) e^{-\nu \tau}+B_\nu(1-e^{-\nu \tau})\\
	&=B_\nu(1-e^{-\nu \tau}).
\end{align*}
Thus, we obtain
\begin{align*}
	V(0)=0,~V'(\tau)&=\nu B_\nu e^{-\nu\tau},~	V''(\tau)= - \nu^2 B_\nu e^{-\nu\tau}~	-\frac{V''(\tau)}{V'(\tau)}=\nu.
\end{align*}

\section{Discussion of Example \ref{par:ex_double_think}} \label{sub:discussion_of_example_par:ex_double_think}

Consider the model of \cref{par:ex_double_think}. Note that whenever $p_\tau$ is such that $\nu=p_\tau \lambda$, the continuation game satisfies our assumptions, and therefore, \Cref{lem:ded} applies to the continuation game. 

Assume that for some time remaining $\tau<T$, the belief is indeed such that $p_\tau \lambda=\nu$ and that $p_\tau<\bar{p}$. If in addition $\gamma_\tau<0$, then we know by the continuity of $\gamma_\tau$ that there is a neighborhood of remaining time $\tau+\varepsilon>\tau$ such that \cref{lem:ded} continues to hold in this neighborhood too.

However, in this neighborhood, $p_{\tau+\varepsilon} \lambda > \nu$, which violates the assumption on relative concavity. Thus, in particular, \cref{lem:minimum} may be violated, which, in turn, implies that---once said neighborhood becomes large---eventually $\gamma_{\tau+\varepsilon}$ may be increasing and may become positive. As a consequence, the agent may engage in an initial thinking period before returning to the path described by \cref{lem:ded}. A numerical solution of \cref{par:ex_double_think} for various deadline lengths is provided in \Cref{fig:doublethink}. As we see, once $T$ is large enough, the optimal policy adds an initial thinking period.

\begin{figure}
\includegraphics[width=.9\textwidth]{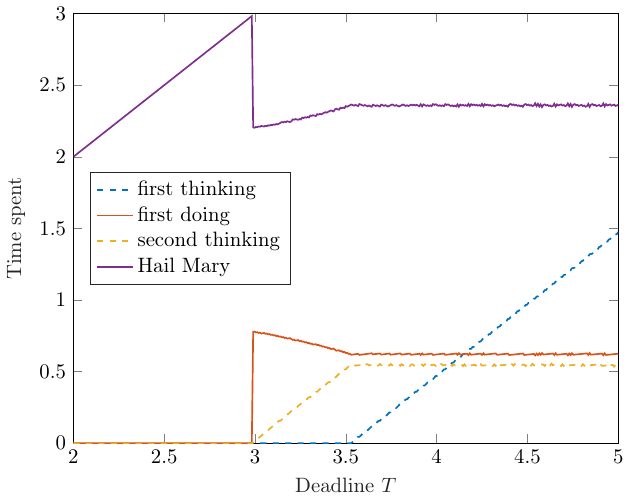}
\caption{\emph{Interval length by deadlines without relative concavity assumption.} The figure shows a numerical solution of the time spent in each interval using a particular approach as a function of the deadline.\newline  Parameters: $B=9,\bar{p}=0.8,c=0.5,\lambda=1,\mu=0.4,\nu=0.5$.}\label{fig:doublethink}
\end{figure}

It is apparent from \Cref{fig:doublethink} that if the time horizon is long enough, then the agent will indeed start by thinking before reverting to our doing-thinking-doing pattern. The reason is as follows: once $\nu<\bar{p} \lambda$ and there is plenty of time remaining, the deadline effect and hence the time pressure is not the agent's primary concern. Instead, in this example, the payoff of successful thinking with a sufficiently long deadline is \emph{higher} than the payoff of successful doing. Thus, with low time pressure, thinking has a payoff advantage over doing and is preferred until the time pressure deteriorates the value of progress on the thinking arm. Once the value of progress is low enough, the intuition and, eventually, the formal analysis of our main model apply again.

\section{Belief-Dependent Continuation Value} \label{sub:belief_V}

In this part, we show how the key lemmata (\Cref{lem:dattheend,lem:evolutiongamma,lem:minimum}) that lead to \Cref{lem:ded} extend to the case in which the continuation value of an arrival on the thinking arm also depends on the belief about the doing arm. Moreover, we verify that \Cref{lem:ded} also holds for the example of a risky new arm without a switching cost.

\paragraph{Modified optimal control problem.}
When the continuation value also depends on $A_\tau$, denote by $V(\tau,A_\tau)$ the value of an arrival on the thinking arm. The Hamiltonian corresponding to the modified optimal control problem becomes
\begin{align*}
	H_\tau(a_\tau;A_\tau):=&e^{-\mu (T- \tau -A_\tau)} (1-\bar{p} + \bar{p}e^{-\lambda A_\tau}) J(p_\tau,\tau,a_\tau) + a_{\tau} \eta_\tau,\\
  = &e^{-\mu (T- \tau -A_\tau)} (1-\bar{p}) \left( (1-a_{\tau})\mu V(\tau,A_\tau) -c\right) \\ &+ e^{-\mu(T-\tau-\mu A_\tau)} \bar{p}e^{-\lambda A_\tau} \left((1-a_\tau)\mu V(\tau,A_\tau) + a_\tau \lambda B-c\right) + a_{\tau} \eta_\tau
\end{align*}
and the co-state evolution becomes
\begin{align*}
	\dot{\eta}_\tau &=e^{-\mu(T-\tau-A_\tau)} (\bar{p} e^{-\lambda A_\tau}+1-\bar{p}) \\
	& ~ \cdot  \Bigg(p_\tau \lambda(\mu - \lambda) a_\tau B + \mu (1-a_\tau) \bigg( (\mu-\lambda p_\tau) V(\tau,A_\tau) + V_A(\tau,A_\tau)\bigg) - 	(\mu-p_\tau \lambda) c \Bigg).
\end{align*}

\paragraph{Corresponding \Cref{lem:dattheend}.} \Cref{lem:dattheend} holds trivially because $\lim_{\tau \rightarrow 0}V(\tau,A_\tau)=0$ in this case as well and the  boundary condition of the optimal control problem is unchanged, $\eta_{\tau=0}=0$.

\paragraph{Corresponding \Cref{lem:evolutiongamma}.}
The resulting switching function is
\begin{align*}{}
	\gamma_\tau = e^{-\mu(T-\tau-A_\tau)}(\bar{p} e^{-\lambda A_\tau}+1-\bar{p})(\mu V(\tau,A_\tau)-p_\tau \lambda B) - \eta_\tau
\end{align*}
with evolution
\begin{align*}
	\dot{\gamma}_\tau=e^{-\mu(T-\tau-A_\tau)}(\bar{p} e^{-\lambda A_\tau}+1-\bar{p}) \bigg(p_\tau \lambda \mu (V(\tau,A_\tau)-B)+\mu(V_\tau(\tau,A_\tau)-V_A(\tau,A_\tau))+(\mu-p_\tau \lambda)c\bigg).
\end{align*}

\paragraph{Corresponding \Cref{lem:minimum}.} 
The derivative with respect to $\tau$ of the analogue of $g(\tau)$ in the switching function, as in \Cref{lem:minimum}, is
\begin{align*}
	\ddot{\gamma}_\tau &= \mu\bigg( \frac{d p_\tau}{d \tau}\lambda \left(B-V(\tau,A_\tau)+\frac{c}{\mu}\right)\\
	&+\lambda p_\tau(V_\tau(\tau,A_\tau)+a_\tau V_A(\tau,A_\tau))+a_\tau (a_\tau V_{A,A}(\tau,A_\tau)+2V_{\tau,A}(\tau,A_\tau)-V_{\tau,\tau}(\tau,A_\tau))\bigg).
\end{align*}
Observe that $\ddot{\gamma}_\tau$ is negative for all $a_\tau$ if $V(\tau,A_\tau)\leq B+\frac{c}{\mu}$ and\footnote{Where the total derivative with respect to time is $\partial_\tau V(\tau, A_\tau)-\partial_A V(\tau,A_\tau)a_\tau$.} \[-\frac{\frac{d^2}{d\tau^2}V(\tau,A_\tau)}{\frac{d}{d\tau}V(\tau,A_\tau)}\geq p_\tau \lambda\] as $\frac{d p_\tau}{d \tau}\geq 0$.

Thus, no interior local minimum can exist under these assumptions, which are analogous to those in the main text.\footnote{Note that we could dispense with the assumption $V(\tau,A_\tau)\leq B+\frac{c}{\mu}$ by plugging $\dot{\gamma}_\tau=0$ into $\ddot{\gamma}_\tau$, as in the proof \Cref{lem:minimum}.}

\paragraph{Corresponding \Cref{lem:ded}.} The proof of the analogous result as in \Cref{lem:ded} follows directly by combining the corresponding \Cref{lem:dattheend,lem:evolutiongamma,lem:minimum}.

\subsection{Verification of Example \ref{par:ex_riskyreturn}} 

We next verify that a simple risky new arm satisfies our assumptions when there is no cost of switching between doing arms and the agent may mix continuously between the two arms. To save on notation and case distinctions, we assume that $\bar{p}^\nu >\bar{p}$ and $\bar{p}\leq 2/3$ and set $\lambda=1$ for both arms.

\Cref{lem:riskythinkingvalue} constructs the reduced form $V(\tau;A_\tau)$ for this example. \Cref{lem:RiskyArmSatisfiesAssumptions} shows that the constructed $V(\tau;A_\tau)$ satisfies the condition required for the Corresponding \Cref{lem:ded}, \(\left(-\frac{d^2}{d\tau^2}V(\tau,A_\tau)\right)/\left(\frac{d}{d\tau}V(\tau,A_\tau)\right)\geq p_\tau \lambda.\)

\begin{lemma}\label{lem:riskythinkingvalue}
	While holding a belief $p_\tau=\bar{p}e^{-\lambda A_\tau}/(\bar{p}e^{-\lambda A_\tau}+ 1-\bar{p})$ on the initial doing arm with time remaining $\tau$, the value of having access to a new risky doing arm with initial belief $\bar{p}^\nu $ is
\begin{align*}
V(\tau,A_\tau)&= \bar{p}^\nu (1-e^{-\min \{\hat{t}(A_\tau),\tau\}})\left(B-c\right)-(1-\bar{p}^\nu ) c\min\{\hat{t}(A_\tau),\tau\} 
	\\&+(1-\bar{p}^\nu )\frac{p_\tau}{1-p_\tau} \left(\left(B-c\right)\left(p_\tau^2(1-e^{- \hat{\tau}})+2p_\tau(1-p_\tau)(1-e^{-\frac{1}{2}\hat{\tau}})\right) \right.\\ 
	&\left.- c\left(2p_\tau(1-p_\tau)(1-e^{-\frac{1}{2}\hat{\tau}})+(1-p_\tau)^2\hat{\tau} \right) \right),
	\end{align*}
	where $\hat{t}(A_\tau)=\ln \left(\frac{\bar{p}^\nu }{1-\bar{p}^\nu }\frac{1-p_\tau}{p_\tau} \right)$, which is the time at which after discovering the new arm, the agent switches from pulling the new arm exclusively to mixing between both arms whenever $\tau>\hat{t}(A_\tau)$. 
\end{lemma}

\begin{proof}
During the time in which the agent exclusively uses the new arm, i.e., for the periods $t\in[0,\min\{\hat{t}(A_\tau),\tau\}]$, the agent's payoff is
\begin{align*}
	V^{new}(\tau;A_\tau):=\bar{p}^\nu  (1-e^{-\min \{\hat{t}(A_\tau),\tau\}})\left(B-c\right)-(1-\bar{p}^\nu ) c\min\{\hat{t}(A_\tau),\tau\}.
\end{align*}
With probability $\frac{1-\bar{p}^\nu }{1-p_\tau}$, the agent does not obtain a success before $\min\{\tau,\hat{t}\}$. In this event, the agent will mix for the remaining time $\hat{\tau}:=\min\{0,\tau-\hat{t}\}$. Because the agent mixes between the two arms instead of using a single arm at a full rate, the beliefs decline at a lower rate on each arm, while the instantaneous success rate at time $\tau$ is still $p_\tau$. In particular, both arms are pulled with the same intensity $a_\tau=\frac{1}{2}$, which implies that 
\begin{align*}
	\dot{p}^\nu_\tau = \dot{p}_\tau = \frac{1}{2}p_\tau(1-p_\tau).
\end{align*}
Solving for the agent's value upon mixing for the remaining time $\hat{\tau}$, we obtain
\begin{align*}
	V^{mix}(\hat{\tau};A_\tau) &:= \int_0^{\hat{\tau}} e^{-p_\tau } \left( p_\tau B -c\right) d\tau\\ 
	&= \left(B-c\right)\left(p_\tau^2(1-e^{- \hat{\tau}})+2p_\tau(1-p_\tau)(1-e^{-\frac{1}{2}\hat{\tau}})\right) \\ 
	&- c\left(2p_\tau(1-p_\tau)(1-e^{-\frac{1}{2}\hat{\tau}})+(1-p_\tau)^2\hat{\tau} \right).
\end{align*}
Putting the pieces together delivers the result.
\end{proof}

Finally, we have to ensure that the relative concavity assumption is satisfied. Note that the maximal time remaining such that the agent uses any of the two arms is bounded from above, as eventually the beliefs would become too low to generate a positive expected payoff. Denote this upper bound, which we explicitly define below, by $\overline{\tau}$.

\begin{lemma}\label{lem:RiskyArmSatisfiesAssumptions}
	Under the assumptions that $\bar{p}^\nu >\bar{p}$, $\bar{p}\leq 2/3$, and $\lambda=1$, $V(\tau,A_\tau)$ satisfies 
	\begin{align*}
		-\frac{\frac{d^2}{d\tau^2}V(\tau,A_\tau)}{\frac{d}{d \tau}V(\tau,A_\tau)}\geq p_\tau.
	\end{align*}
\end{lemma}

\begin{proof}
Note that varying the time remaining has different effects on $V(\tau,A_\tau)$ depending on whether the agent thinks or does. If she does, $A_\tau$ and $\tau$ vary both, affecting the value of a thinking success. If she thinks, only the change in $\tau$ affects the value of a thinking success.

To verify our assumptions, note that\footnote{We simplify notation by supressing arguments whenever it should not cause confusion. Moreover, we use the belief as the state variable, which is equivalent to using $A_\tau$.} 
\begin{align*}
	\frac{d}{d \tau}V(\tau,p_\tau)=&\frac{\partial{V}}{\partial{\tau}}+ \dot{p}_\tau \frac{\partial{V}}{\partial{p}} \\
	\frac{d^2}{d\tau^2}V(\tau,p_\tau)&= \frac{\partial^2{V}}{(\partial \tau)^2} + 2 \frac{\partial^2{V}}{\partial\tau \partial p} \dot{p}_\tau + \dot{p}_\tau^2 \frac{\partial^2{V}}{(\partial p)^2} + \ddot{p}_\tau \frac{\partial{V}}{\partial{p}}.
\end{align*}
Moreover, $V(\tau,p_\tau)=V^{new} + \frac{1-\bar{p}^\nu }{1-p_\tau} V^{mix}$. 

If the deadline is too close when thinking is successful, then we are in the case of \cref{par:ex_risky}, and our assumptions are satisfied. The domain of time remaining under which mixing is relevant is $\tau\in [\underline{\tau},\overline{\tau}]$, where $\overline{\tau}:=\ln \Big(\frac{\bar{p}^\nu }{1-\bar{p}^\nu }\frac{1-p_\tau}{p_\tau} \Big)+ 2\ln\Big( \frac{p_\tau}{1-p_\tau}\frac{ B- c}{c}\Big)$ is the time at which the agent would prefer shirking over working on either, as both beliefs have declined too much, and where $\underline{\tau}:=\ln \Big(\frac{\bar{p}^\nu }{1-\bar{p}^\nu }\frac{1-p_\tau}{p_\tau} \Big)$ is the time at which the belief about the new arm has declined to the current belief of the doing arm.

To simplify the notation, define $V^{p-mix}:=\frac{1-\bar{p}^\nu }{1-p_\tau}V^{mix}$, and note that for $b,d>0$, if $\frac{a}{b} \geq x$ and $\frac{c}{d}\geq x$, then $\frac{a+c}{b+d} \geq x$.\footnote{Observe that $\frac{a}{b} \geq x$ and that $\frac{c}{d} \geq x$ imply $a>bx$ and $c>dx$; thus, $a+c>(b+d)x$.} Thus, it is sufficient to show that (a) $\frac{-\frac{d^2}{d\tau^2}V^{new}}{\frac{d}{d\tau}V^{{new}}}>p_\tau$ and (b) $\frac{-\frac{d^2}{d\tau^2}V^{p-mix}}{\frac{d}{d\tau}V^{p-mix}}\geq p_\tau$. 

It is straightforward to see that (a) is satisfied as
\begin{align*}
	\frac{\frac{-d^2 V^{new}}{d \tau^2}}{\frac{d V^{new}}{d \tau}}&= \frac{p_\tau(B-c)\frac{1-\bar{p}^\nu }{1-p_\tau}}{(p_\tau B-c)\frac{1-\bar{p}^\nu }{1-p_\tau}}
	>p_\tau.
\end{align*}
Next, we show (b), i.e., that $\frac{-\frac{d^2}{d\tau^2}V^{p-mix}}{\frac{d}{d\tau}V^{p-mix}}\geq p_\tau $. Note that both the numerator and denominator are positive. Hence, a lower bound for the fraction is given by dividing the lower bound of $-\frac{d^2}{d\tau^2}V^{p-mix}$ by the upper bound of $\frac{d}{d\tau}V^{p-mix}$. We obtain the bounds by showing that for any feasible parameter constellation, (i) the smallest numerator is attained for $\overline{\tau}$.\footnote{This is the upper bound on the deadline such that the agent is willing to exert effort for any parameters and time remaining.} and (ii) the greatest numerator is attained for $\underline{\tau}.$\footnote{This is the lower bound on the deadline such that the mixing phase is reached.}

To see (i), observe that 
\begin{align*}
	\frac{d^2}{d \tau^2} \left(-\frac{d^2}{d\tau^2}V^{p-mix}\right) &= - \frac{p_\tau \bar{p}^\nu  }{\sqrt{\frac{\bar{p}^\nu }{1-\bar{p}^\nu }\frac{1-p_\tau}{p_\tau}}}\left(( B - 2c) e^{\frac{\tau}{2}}(1-p_\tau)^2+( B -2c)p_\tau(3-2p_\tau)\sqrt{\frac{\bar{p}^\nu }{1-\bar{p}^\nu }\frac{1-p_\tau}{p_\tau}} \right)\\&<0.
\end{align*}
Hence, the derivative of $\left(-\frac{d^2}{d\tau^2}V^{p-mix}\right)$ is decreasing. Next, note that 
\begin{align*}
	\lim_{\tau\rightarrow \overline{\tau}}\frac{d}{d \tau} \left(-\frac{d^2}{d\tau^2}V^{p-mix}\right) = (1-\bar{p}^\nu )(1-p_\tau)p_\tau  \frac{p_\tau B c}{ B - c}>0
\end{align*}
and therefore that $-\frac{d^2}{d\tau^2}V^{p-mix}$ is increasing for all $\tau$ on the relevant domain. Hence, its lower bound is attained at $\tau=\underline{\tau}$ and is given by
\begin{align*}
	\lim_{\tau\rightarrow \underline{\tau}} \left(-\frac{d^2}{d\tau^2}V^{p-mix}\right) &= \frac{1-\bar{p}^\nu }{1-p_\tau}\frac{p_\tau }{ B -c} \left(B^2  (p_\tau (p_\tau (2 p_\tau-5)+4)-2)+2 B c   (3-2 p_\tau) p_\tau^2 \right. \\ & \left. -2 c (1-p_\tau)^2 (1-2 p_\tau) (B-c)
   \ln \left(\frac{p_\tau (B-c  )}{c (1-p_\tau)}\right)-c^2 (4 (2-p_\tau) p_\tau-3)\right).
\end{align*}
To see (ii), observe that 
\begin{align*}
	\frac{d^2}{d \tau^2} \left(\frac{d}{d\tau}V^{p-mix}\right) = -\frac{ p_\tau \bar{p}^\nu  }{2} e^{-   \tau} \left(\frac{(p_\tau-1) e^{\frac{  \tau}{2}} (B  -2 c)}{\sqrt{\frac{(p_\tau-1)
   \bar{p}^\nu }{p_\tau (\bar{p}^\nu -1)}}}-2 B   p_\tau+2 c p_\tau\right)>0
\end{align*}
which implies that $\frac{d}{d \tau} \left(\frac{d}{d\tau}V^{p-mix}\right) $ is increasing. At the upper bound, $ \lim_{\tau\rightarrow \overline{\tau}} \frac{d^2}{d \tau^2} \left(\frac{d}{d\tau}V^{p-mix}\right) =0$. Hence, the first derivative of $\frac{d}{d\tau}V^{p-mix}$ with respect to $\tau$ is negative throughout, and the upper bound is attained at $\tau\rightarrow \underline{\tau}$, with
\begin{align*}
	\lim_{\tau\rightarrow \underline{\tau}} \frac{d}{d\tau}V^{p-mix} = \frac{2 (1-\bar{p}^\nu ) (B   p_\tau-c)}{1-p_\tau}>0.
\end{align*}

We therefore obtain that
\begin{align*}
	&\frac{-\frac{d^2}{d\tau^2}V^{p-mix}}{\frac{d}{d\tau}V^{p-mix}}\\
	 &> \frac{ B^2  (p_\tau (p_\tau (2 p_\tau-5)+4)-2)+2 B c   (3-2 p_\tau) p_\tau^2}{B c  (1-p_\tau)^2 } \\
	 &- \frac{2  (1-p_\tau)^2 (1-2 p_\tau) (B-c)
   \ln \left(\frac{p_\tau (B-c)}{c (1-p_\tau)}\right)-c (4 (2-p_\tau-2) p_\tau-3)}{B c  (1-p_\tau)^2 }.
\end{align*}

Finally, we need to verify that the right-hand side of the last expression is greater than $p_\tau $. To see this, we compute a lower bound of it using the fact that it is decreasing in $c$ (see below). As mixing requires that $p_\tau  B \geq c$, a lower bound is attained for $c=p_\tau  B$. 

To see that the term is decreasing in $c$, observe that it is convex in $c$, as its second derivative is
\begin{align*}
	\frac{2 B (B   (p_\tau (p_\tau (5-2 p_\tau)-4)+2)-c)}{c^3 (1-p_\tau)^2 (B-c)}
\end{align*}
which is positive, as $ B>2c$ and $p_\tau<2/3$. Thus, the first derivative is increasing, and at the upper bound of $c$, it reduces to
\begin{align*}
	-\frac{2}{B(1-p_\tau)p_\tau^2 }<0.
\end{align*}
Hence, a lower bound of the fraction under consideration is attained for $c=p_\tau  B$, which is
\(2/p_\tau\), which is strictly larger than $p_\tau$.
\end{proof}
 
\section{No-Shirking Condition} \label{sub:no_shirking_conditions}

In the text, we assume that $B$ is high enough such that the agent never shirks if she has not yet found a solution. Here, we show that such a $B$ always exists and is finite. Moreover, we provide an (implicit) construction. It is sufficient to show that the agent has an incentive to pull an arm at the deadline. The agent does not shirk if her terminal belief $\underline{p}\geq\frac{c}{\lambda B}$. For any $\underline{p}>0$ and $T<\infty$, there is a $\underline{B}<\infty$ such that the above condition holds for any $B\geq\underline{B}$. For any $T<\infty$, the terminal belief is weakly larger than 

\[p^{\min}= \frac{\bar{p} e^{-\lambda T}}{\bar{p} e^{-\lambda T}+1-\bar{p}}>0.\]
By \cref{lem:thinkingbetterwithlonghorizons}, for any $B$, there is a $\bar{T}<\infty$ such that $\underline{p}>p^{min}$ for $T>\bar{T}$. Thus, there is a $\underline{B}<\infty$ such that the agent never shirks for \emph{any} deadline, including the limit $T\rightarrow\infty$.

 \end{document}